\newif\ifprocs
\procstrue
\procsfalse 

\newif\ifarxiv
\arxivtrue

\newif\ifcomments
\commentstrue
\commentsfalse 

\ifprocs

\documentclass[twoside,leqno]{article}

\usepackage[letterpaper]{geometry}
\usepackage{amsmath}
\usepackage{xspace}
\usepackage{fancyhdr}

\usepackage{ltexpprt}
\usepackage{hyperref}
\usepackage[export]{adjustbox}
\usepackage{thmtools}
\usepackage{amssymb}

\def\compactify{\itemsep=0pt \topsep=0pt \partopsep=0pt \parsep=0pt}

\newcommand{\ProblemName}[1]{\textsf{#1}}
\newcommand{\MF}{\ProblemName{Max-Flow}\xspace}
\newcommand{\MC}{\ProblemName{Min-Cut}\xspace}
\newcommand{\GMC}{\ProblemName{Global-Min-Cut}\xspace}
\newcommand{\EC}{\ProblemName{Edge-Connectivity}\xspace}
\newcommand{\stMC}{\ProblemName{st-Min-Cut}\xspace}
\newcommand{\LMC}{\ProblemName{Latest Min-Cut}\xspace}
\newcommand{\LMwC}{\ProblemName{Latest Min$_{2w}$-Cut}\xspace}
\newcommand{\expanders}{\textsc{Expanders-Guided Querying}\xspace}
\newcommand{\pC}{\textsc{Isolating-Cuts}\xspace}

\newcommand{\SSMF}{\ProblemName{Single-Source Max-Flow}\xspace}
\newcommand{\APMF}{\ProblemName{All-Pairs Max-Flow}\xspace}
\newcommand{\APEC}{\ProblemName{All-Pairs Edge Connectivity}\xspace}
\newcommand{\APSP}{\ProblemName{All-Pairs Shortest-Path}\xspace}
\newcommand{\SPath}{\ProblemName{Shortest-Path}\xspace}
\newcommand{\APR}{\ProblemName{All-Pairs Reachability}\xspace}
\newcommand{\STMF}{\ProblemName{ST-Max-Flow}\xspace}
\newcommand{\GMF}{\ProblemName{Global Max-Flow}\xspace}
\newcommand{\MFV}{\ProblemName{Max-Flow}\xspace}

\newcommand{\GHU}{\ProblemName{Gomory-Hu}\xspace}
\newcommand{\lucky}{isolated\xspace}
\newcommand{\GHT}{Gomory-Hu tree\xspace}
\newcommand{\TOV}{\ProblemName{$3$OV}\xspace}
\newcommand{\CAG}{\ProblemName{CAG}\xspace}
\newcommand{\LCA}{\ProblemName{LCA}\xspace}
\newcommand{\GHEPT}{\ProblemName{GH-Equivalent Partition Tree}\xspace}
\newcommand{\PT}{\ProblemName{Partition Tree}\xspace}
\newcommand{\dem}{\mathbf{d}\xspace}

\newcommand{\CAGs}{{\CAG}s\xspace}
\newcommand{\AMC}{\ProblemName{Approx-Min-Cut}\xspace}

\DeclareMathOperator{\SOL}{SOL}
\DeclareMathOperator{\dist}{dist}
\DeclareMathOperator{\wdeg}{wdeg}
\DeclareMathOperator{\TTIME}{TIME}
\DeclareMathOperator{\FirstVisit}{FirstVisit}
\DeclareMathOperator{\LastVisit}{LastVisit} 
\DeclareMathOperator{\parent}{parent} 
\DeclareMathOperator{\lbl}{label} 
\DeclareMathOperator{\private}{private} 
\DeclareMathOperator{\vol}{vol} 
\DeclareMathOperator{\size}{size}

\let\poly\relax
\DeclareMathOperator{\poly}{poly} 

\newcommand{\OutputLen}{\text{output}}

\newcommand{\MMOD}[1]{\ (\mathrm{mod}\ #1)}

\newcommand\eps{\varepsilon}
\renewcommand\epsilon{\varepsilon}
\newcommand\tO{\ensuremath{\tilde O}}
\newcommand{\Madry}{M{\k{a}}dry\xspace}
\newcommand{\Raecke}{R\"{a}cke\xspace} 
\newcommand{\calF}{\mathcal{F}}
\newcommand{\T}{\mathcal{T}}
\newcommand{\TG}{\mathcal{T}^*} 
\DeclareMathOperator{\st}{st}
\DeclareMathOperator{\fn}{fn}
\providecommand{\set}[1]{{\{#1\}}}
\providecommand{\card}[1]{\lvert#1\rvert}
\providecommand{\minn}[1]{\min\set{#1}}
\providecommand{\maxx}[1]{\max\set{#1}}

\begin{document}

\newcommand\relatedversion{}

\title{\Large Friendly Cut Sparsifiers and Faster Gomory-Hu Trees\relatedversion}

\author{Amir Abboud%
  \thanks{Work partially supported by an Alon scholarship and a research grant from the Center for New Scientists at the Weizmann Institute of Science.
    Email: \texttt{amir.abboud@weizmann.ac.il}
  }
  \\ Weizmann Institute of Science
  \\ Israel  
  \and Robert Krauthgamer%
  \thanks{Work partially supported by ONR Award N00014-18-1-2364,
    the Israel Science Foundation grant \#1086/18,
    and a Minerva Foundation grant.
    Email: \texttt{robert.krauthgamer@weizmann.ac.il}
  }
  \\ Weizmann Institute of Science
  \\ Israel
  \and Ohad Trabelsi%
  \thanks{Work partially done at Weizmann Institute of Science and partially supported by the NSF Grant CCF-1815316.
    Email: \texttt{ohadt@umich.edu}
  }
  \\University of Michigan
  \\USA
}

\date{}

\maketitle


\fancyfoot[R]{\scriptsize{Copyright \textcopyright\ 2022 by SIAM\\
Unauthorized reproduction of this article is prohibited}}





\else 

\documentclass[11pt,USenglish]{article}
\usepackage[margin=1.0in]{geometry}
\usepackage[T1]{fontenc}
\usepackage[utf8]{inputenc}
\usepackage[english]{babel}
\usepackage{authblk}
\usepackage{chngcntr}
\counterwithin{equation}{section}
\usepackage{amssymb}

\usepackage[normalem]{ulem}
\usepackage{enumerate}
\usepackage{mwe}
\usepackage{caption}
\usepackage{xcolor}
\usepackage{algorithm}
\usepackage[noend]{algpseudocode}
\renewcommand{\algorithmicrequire}{\textbf{Input:}} 
\renewcommand{\algorithmicensure}{\textbf{Output:}} 

\usepackage{amsmath}
\usepackage{graphicx}

\ifarxiv\else
\graphicspath{{Figures/}}
\fi 

\usepackage[export]{adjustbox}
\usepackage{thmtools}
\usepackage{amsthm}
\usepackage{thm-restate}
\usepackage{xspace}

\usepackage[colorlinks,linkcolor=blue,citecolor=blue,filecolor=blue,urlcolor=blue]{hyperref}

\newtheorem{infthm}{Informal Theorem}
\newtheorem*{restatement}{Restatement of }
\newtheorem{oq}{Open Question}

\newcommand{\defproblem}[2]{
\bigskip
\begin{center}
\noindent\fbox{
	\begin{minipage}{.96\linewidth}
	\textbf{#1}
	
	\smallskip
		{#2}
	\end{minipage}
}
\end{center}
\medskip
}

\newtheorem{theorem}{Theorem}[section]
\newtheorem{lemma}[theorem]{Lemma}
\newtheorem{assumption}[theorem]{Assumption}
\newtheorem{definition}[theorem]{Definition}
\newtheorem{open}[theorem]{Open Problem}
\newtheorem{openq}[theorem]{Open Question}
\newtheorem{observation}[theorem]{Observation}
\newtheorem{fact}[theorem]{Fact}
\newtheorem{corollary}[theorem]{Corollary}
\newtheorem{hypothesis}[theorem]{Hypothesis}
\newtheorem{proposition}[theorem]{Proposition}

\theoremstyle{plain}
\newtheorem{claim}[theorem]{Claim}
\newtheorem{conjecture}[theorem]{Conjecture}

\makeatletter
\newtheorem*{rep@theorem}{\rep@title}
\newcommand{\newreptheorem}[2]{%
\newenvironment{rep#1}[1]{%
 \def\rep@title{#2 \ref{##1}}%
 \begin{rep@theorem}}%
 {\end{rep@theorem}}}
\makeatother
\newreptheorem{theorem}{Theorem}

\makeatletter
\newtheorem*{rep@corollary}{\rep@title}
\newcommand{\newrepcorollary}[2]{%
\newenvironment{rep#1}[1]{%
 \def\rep@title{#2 \ref{##1}}%
 \begin{rep@corollary}}%
 {\end{rep@corollary}}}
\makeatother
\newreptheorem{corollary}{Corollary}

\newreptheorem{lemma}{Lemma}

\def\compactify{\itemsep=0pt \topsep=0pt \partopsep=0pt \parsep=0pt}

\newcommand{\colnote}[3]{\textcolor{#1}{$\ll$\textsf{#2}$\gg$}}

\ifcomments
\newcommand{\rnote}[1]{\colnote{red}{#1--Robi}{RK}}
\newcommand{\anote}[1]{\colnote{olive}{#1--Amir}{AA}}
\newcommand{\onote}[1]{\colnote{blue}{#1--Ohad}{OT}}
\else
\newcommand{\rnote}[1]{}
\newcommand{\anote}[1]{}
\newcommand{\onote}[1]{}
\fi 

\newcommand{\ProblemName}[1]{\textsf{#1}}
\newcommand{\MF}{\ProblemName{Max-Flow}\xspace}
\newcommand{\MC}{\ProblemName{Min-Cut}\xspace}
\newcommand{\GMC}{\ProblemName{Global-Min-Cut}\xspace}
\newcommand{\EC}{\ProblemName{Edge-Connectivity}\xspace}
\newcommand{\stMC}{\ProblemName{st-Min-Cut}\xspace}
\newcommand{\LMC}{\ProblemName{Latest Min-Cut}\xspace}
\newcommand{\LMwC}{\ProblemName{Latest Min$_{2w}$-Cut}\xspace}
\newcommand{\expanders}{\textsc{Expanders-Guided Querying}\xspace}
\newcommand{\pC}{\textsc{Isolating-Cuts}\xspace}

\newcommand{\SSMF}{\ProblemName{Single-Source Max-Flow}\xspace}
\newcommand{\APMF}{\ProblemName{All-Pairs Max-Flow}\xspace}
\newcommand{\APEC}{\ProblemName{All-Pairs Edge Connectivity}\xspace}
\newcommand{\APSP}{\ProblemName{All-Pairs Shortest-Path}\xspace}
\newcommand{\SPath}{\ProblemName{Shortest-Path}\xspace}
\newcommand{\APR}{\ProblemName{All-Pairs Reachability}\xspace}
\newcommand{\STMF}{\ProblemName{ST-Max-Flow}\xspace}
\newcommand{\GMF}{\ProblemName{Global Max-Flow}\xspace}
\newcommand{\MFV}{\ProblemName{Max-Flow}\xspace}

\newcommand{\GHU}{\ProblemName{Gomory-Hu}\xspace}
\newcommand{\lucky}{isolated\xspace}
\newcommand{\GHT}{Gomory-Hu tree\xspace}
\newcommand{\TOV}{\ProblemName{$3$OV}\xspace}
\newcommand{\CAG}{\ProblemName{CAG}\xspace}
\newcommand{\LCA}{\ProblemName{LCA}\xspace}
\newcommand{\GHEPT}{\ProblemName{GH-Equivalent Partition Tree}\xspace}
\newcommand{\PT}{\ProblemName{Partition Tree}\xspace}
\newcommand{\dem}{\mathbf{d}\xspace}

\newcommand{\CAGs}{{\CAG}s\xspace}
\newcommand{\AMC}{\ProblemName{Approx-Min-Cut}\xspace}

\DeclareMathOperator{\SOL}{SOL}
\DeclareMathOperator{\dist}{dist}
\DeclareMathOperator{\wdeg}{wdeg}
\DeclareMathOperator{\TTIME}{TIME}
\DeclareMathOperator{\FirstVisit}{FirstVisit}
\DeclareMathOperator{\LastVisit}{LastVisit} 
\DeclareMathOperator{\parent}{parent} 
\DeclareMathOperator{\lbl}{label} 
\DeclareMathOperator{\private}{private} 
\DeclareMathOperator{\vol}{vol} 
\DeclareMathOperator{\size}{size}

\let\poly\relax
\DeclareMathOperator{\poly}{poly} 

\newcommand{\OutputLen}{\text{output}}

\newcommand{\MMOD}[1]{\ (\mathrm{mod}\ #1)}

\newcommand\eps{\varepsilon}
\renewcommand\epsilon{\varepsilon}
\newcommand\tO{\ensuremath{\tilde O}}
\newcommand{\Madry}{M{\k{a}}dry\xspace}
\newcommand{\Raecke}{R\"{a}cke\xspace} 
\newcommand{\calF}{\mathcal{F}}
\newcommand{\T}{\mathcal{T}}
\newcommand{\TG}{\mathcal{T}^*} 
\DeclareMathOperator{\st}{st}
\DeclareMathOperator{\fn}{fn}
\providecommand{\set}[1]{{\{#1\}}}
\providecommand{\card}[1]{\lvert#1\rvert}
\providecommand{\minn}[1]{\min\set{#1}}
\providecommand{\maxx}[1]{\max\set{#1}}


\begin{document}

\title{Friendly Cut Sparsifiers and Faster Gomory-Hu Trees}
\date{}

\author[1]{Amir Abboud%
  \thanks{Work partially supported by an Alon scholarship and a research grant from the Center for New Scientists at the Weizmann Institute of Science.
  Email: \texttt{amir.abboud@weizmann.ac.il}
  }}
\author[2]{Robert Krauthgamer\thanks{Work partially supported by ONR Award N00014-18-1-2364,
    the Israel Science Foundation grant \#1086/18,
    and a Minerva Foundation grant.
    Email: \texttt{robert.krauthgamer@weizmann.ac.il}
  }}
\author[3]{Ohad Trabelsi%
  \thanks{Work partially done at Weizmann Institute of Science and partially supported by the NSF Grant CCF-1815316.
    Email: \texttt{ohadt@umich.edu}
  }}
\affil[1]{Weizmann Institute of Science, Israel}
\affil[2]{Weizmann Institute of Science, Israel}
\affil[3]{University of Michigan, USA}

\maketitle
\thispagestyle{empty}
\setcounter{page}{0}

\fi 

\begin{abstract}
We devise new cut sparsifiers that are related to 
the classical sparsification of Nagamochi and Ibaraki [Algorithmica, 1992],
which is an algorithm that, given an unweighted graph $G$ on $n$ nodes and a parameter $k$,
computes a subgraph with $O(nk)$ edges that preserves all cuts of value up to $k$.
We put forward the notion of a \emph{friendly cut sparsifier}, 
which is a minor of $G$ that preserves all friendly cuts of value up to $k$,
where a cut in $G$ is called \emph{friendly} if every node has more edges
connecting it to its own side of the cut than to the other side.
We present an algorithm that, given a simple graph $G$,
computes in almost-linear time 
a friendly cut sparsifier with $\tilde{O}(n \sqrt{k})$ edges.
Using similar techniques,
we also show how, given in addition a terminal set $T$,
one can compute in almost-linear time a \emph{terminal sparsifier},
which preserves the minimum $st$-cut between every pair of terminals,
with $\tilde{O}(n \sqrt{k} + |T| k)$ edges. 

Plugging these sparsifiers into the recent $n^{2+o(1)}$-time algorithms for constructing a Gomory-Hu tree of simple graphs,
along with a relatively simple procedure for handling the unfriendly minimum cuts,
we improve the running time for moderately dense graphs
(e.g., with $m=n^{1.75}$ edges).
In particular, assuming a linear-time Max-Flow algorithm,
the new state-of-the-art for Gomory-Hu tree is
the minimum between our $(m+n^{1.75})^{1+o(1)}$ and the known $m n^{1/2+o(1)}$.

We further investigate the limits of this approach and the possibility of better sparsification. 
Under the hypothesis that an $\tilde{O}(n)$-edge sparsifier that preserves all friendly \emph{minimum} $st$-cuts can be computed efficiently, our upper bound improves to $\tilde{O}(m+n^{1.5})$ which is the best possible without breaking the cubic barrier for constructing Gomory-Hu trees in non-simple graphs.
\end{abstract}


\section{Introduction}
\label{sec:intro}

This paper is on the rich and vibrant topic of graph sparsification, where the aim is to reduce the size of the graph, typically measured by the number of edges, while preserving the graph's properties as much as possible. 
This notion is appealing in computer science due to the gains in efficiency, across all metrics, both in theory and in practice, that come from working with smaller objects.
In particular, it is a critical step inside state-of-the-art algorithms for many fundamental problems.

One of the most influential results on sparsification, due to Nagamochi and Ibaraki~\cite{NI92}, essentially says that a set of $w$ maximally spanning forests of a graph preserves all cuts with $\leq w$ edges.
Its applications for speeding up algorithms are far-reaching and include problems like minimum cut \cite{SW97}, traveling salesman \cite{junger1995traveling}, disjoint paths \cite{KKR12}, and most importantly for this paper, cut-equivalent (aka Gomory-Hu) trees.

\begin{theorem}[Nagamochi-Ibaraki Sparsification~\cite{NI92}]
Given an unweighted graph $G$ on $n$ nodes and $m$ edges,
one can compute in $O(m)$-time
a subgraph with $O(nw)$ edges that preserves all cuts of value $\leq w$.
\end{theorem}

This sparsification tool is indispensable for the recent algorithms that break the longstanding cubic barrier for the classical \GHT problem.
Gomory and Hu \cite{GH61} discovered that every graph $G$ has a cut-equivalent (weighted) tree $T$ on the same set of nodes, such that for all $s,t \in V(G)$ the minimum $s,t$-cut in $T$ is also the minimum in $G$. 
Moreover, they showed that this tree can be computed using $n-1$ calls to a \MF algorithm, and a major open question since then has been to determine the time complexity of computing such a tree.
The quest for faster algorithms for \GHT is further motivated by the fact that it is a near-optimal data structure for answering minimum $s,t$-cut queries \cite{AKT20}.
Only last year, a subcubic \GHT algorithm for unweighted (simple) graphs
was discovered by Abboud, Krauthgamer, and Trabelsi \cite{AKT21_stoc}.
And then in three independent follow-up papers, by those same authors \cite{AKT21_focs}, by Li, Panigrahi, and Saranurak \cite{LPS21_focs}, and by Zhang \cite{Zhang21}, the bound was brought down from $n^{2.5+o(1)}$ to $n^{2+o(1)}$, which is clearly almost-optimal for dense graphs.%
\footnote{Zhang's algorithm has better running time $\tilde{O}(n^2)$, but requires a (hypothetical) $\tilde{O}(m)$-time \MF algorithm.}

At a very high level, the way that Nagamochi-Ibaraki sparsification fits into the \GHT algorithms is as follows.
These algorithms reduce the problem to $(n/w)^{1+o(1)}$ calls to a \MF algorithm on graphs where only cuts of value in the range $[w,2w]$ are of interest, for $w=2^0,2^1,\ldots,2^{\log n}$.
Instead of making these calls on the input graph, which might require $n/w \cdot \Omega(n^2) = \Omega(n^3)$ time, the algorithms operate on a Nagamochi-Ibaraki sparsifier that has only $O(nw)$ edges, resulting in a time bound of $(n/w)^{1+o(1)} \cdot (nw)^{1+o(1)} = n^{2+o(1)}$,
if \MF is solvable in almost-linear time (and it turns out that the existing algorithms, due to \cite{linearflow21}, are sufficient).

We aim to investigate the possibility of better sparsification methods for the \GHT problem and (hopefully) beyond.
While the most outstanding open question in this context is whether subcubic time is possible for \emph{weighted} graphs, the story for unweighted (simple) graphs remains unfinished:
\emph{What is the time complexity if the number of edges $m$ is taken into account?} 
The state of the art is $\min\{ n^{2+o(1)}, m n^{1/2+o(1)} \}$ \cite{LPS21_focs}, and the only lower bound other than $\Omega(m)$ is the observation in \cite{AKT21_focs} that $\Omega(n^{1.5})$ time is required, even when $m=n$,
unless the cubic barrier can be broken for multigraphs (which seems to be the main roadblock towards weighted graphs).
This leaves a gap of up to $\sqrt{n}$ and lets one hope for the conditionally-best-possible $O(m+n^{1.5})$ bound;
can we achieve it with better sparsifiers?



\paragraph{Better sparsification?}

Alas, it is easy to see that $\Omega(nw)$ is a lower bound on any sparsifier that preserves all cuts with $\leq w$ edges.
And even though we are only interested in \emph{minimum $s,t$-cuts} of value $\leq w$ rather than $\emph{all}$ cuts, the $\Omega(nw)$ barrier remains.
This would even be the case if the sparsifier is allowed to contract vertices (which makes sense if there is no cut of value $\leq w$ separating them) in addition to deleting edges.

\begin{definition}[Sparsifiers]
A \emph{sparsifier} of $G$ is a graph $H$ obtained from it by deleting edges and contracting subsets of the nodes. 
A \emph{contracted graph} is a special case of a sparsifier that is obtained only by contracting subsets of the nodes (and then removing any self-loops).%
\footnote{Recall that a \emph{minor} is obtained by deleting edges and contracting \emph{connected} subsets of nodes. This distinction is not relevant to our results, which extend immediately to minors. Indeed, the number of edges in our sparsifiers would not change if a contraction of a set $S$ would be replaced by contracting every connected subset of $S$. 
}
We say that $H$ is a \emph{$w$-cut sparsifier} of $G$ if it preserves all cuts with $\leq w$ edges in $G$, meaning that none of their edges is deleted and nodes from different sides of such a cut are not contracted together.
\end{definition}

We usually refer to the \emph{size} of a sparsifier as the number of edges in it, counting parallel edges separately.
For many applications it makes sense to combine parallel edges into one weighted edge and count it only once, but this will not be used in this paper.  

\begin{observation}[Optimality of Nagamochi-Ibaraki]
\label{obs:clique}
For every $n\geq w\geq 1$,
there exists a simple graph on $n$ nodes where every $w$-cut sparsifier must have $\Omega(nw)$ edges.
\end{observation}

\begin{proof}
For simplicity, consider $w=n$ and take a clique on $n$ nodes.
(The proof generalizes for all $w$ by taking $n/w$ disjoint $w$-cliques.) 
The minimum $s,t$-cut for any pair $s,t$ is $\{s\}$ or $\{ t \}$ and it has $n-1<w$ edges, so they must all be preserved.
Deleting any edge $\{u,v\}$ decreases the value of the minimum $u,v$-cut, and contracting it increases their minimum cut (to infinity).
\end{proof}

Fortunately, in many contexts, cliques and other structures that prevent us from beating the $O(nw)$ bound actually make the downstream problem easy.
Conceptually, our message is that better sparsification may be possible
if the goal is relaxed to preserving only the ``hard'' cuts.
In the context of \GHT (and perhaps elsewhere), the hard cuts are friendly cuts, as discussed next,
and our main point is that these can be preserved using $o(nw)$ edges.

\subsection{Friendly Cut Sparsifiers}
Given a cut, a node is called friendly if it has more neighbors to its own side than to the other side of the cut.
The definition can be generalized so that $\alpha$-friendly means that more than(or at least) an $\alpha$-fraction of the neighbors are on the same side,
for $\alpha$ that is not necessarily $1/2$.
The cut is then called $\alpha$-friendly if all nodes are $\alpha$-friendly. 
Throughout this paper, we use $\alpha=0.4$ for simplicity; our sparsification results will actually hold for any $\alpha$ but the application to \GHT algorithms will require a fixed $\alpha<0.5$.
Thus, for a cut to be unfriendly, there needs to be a node with $>60 \%$ of its edges going to the other side.
Throughout, for a node $v$ in a graph $G$ we denote by $\deg_G(v)$ the total degree of $v$ in $G$ counting multiple edges accordingly, and omitting the subscript when it is clear from the context.

\begin{definition}[Friendly and Unfriendly Cuts]
A cut $S$ in a graph $G=(V,E)$ is called \emph{unfriendly} if
there exists a node $s \in S$ such that $|E(\{s\},V\setminus S)| > 0.6 \deg(s)$ or there exists a node $t \notin S$ such that $|E(\{t\}, S)| > 0.6 \deg(t)$.
Otherwise, the cut is called \emph{friendly}.
\end{definition}

Questions involving friendly cuts are extensively studied in combinatorics (see \cite{BanLinial16,Ferber21} and references therein). They actually arise in various contexts, including social learning, set theory, and statistical physics, and under different names, such as \emph{internal partitions} \cite{BanLinial16}, \emph{satisfactory partitions} \cite{GerberKobler00,BTV06,GMT20}, and \emph{$q$-cohesive sets} \cite{Morris00}.



Unfriendly cuts can be viewed as a generalization of the \emph{trivial cuts} $(\{v\}, V \setminus \{v\})$, also known as \emph{singleton-} or \emph{degree-cuts}, that happen to be the minimum $s,t$-cuts in a clique.
Observe that a trivial cut is not $\alpha$-friendly for any $\alpha>0$.
So perhaps the lower bound of Observation~\ref{obs:clique} can be bypassed if one only wishes to preserve the friendly cuts?

\begin{definition}[Friendly Cut Sparsifiers]
A \emph{friendly $w$-cut sparsifier} of $G$ is a graph $H$
that preserves all the \emph{friendly} cuts with $\leq w$ edges in $G$.
\end{definition}

Indeed, the main result of this paper is a \emph{friendly cut sparsifier} giving a polynomial improvement over the Nagamochi-Ibaraki bound. The proof is in Section~\ref{sec:friendly}.

\begin{theorem}[Friendly Cut Sparsifiers]
\label{thm:friendly}
There is a randomized $m^{1+o(1)}$-time algorithm that, given a simple graph on $n$ nodes and $m$ edges and a parameter $w$,
produces a friendly $w$-cut sparsifier with $m'=O(n\sqrt{w} \log^{4}{n})$ edges.
There is also a deterministic $m^{1+o(1)}$-time algorithm achieving $m'=(n\sqrt{w})^{1+o(1)}$.
Furthermore, the sparsifier is a contracted graph (and a minor).
\end{theorem}

Technically, the result is very different from Nagamochi and Ibaraki's, and is based on an expander decomposition of the graph rather than on removing spanning forests.
Perhaps the most similar result of this kind is by Kawarabayashi and Thorup \cite{KThorup19}, who study deterministic algorithms for edge connectivity. 
They show that all non-trivial minimum cuts in a simple graph of minimum degree $\delta$ can be preserved by a contracted graph on $\tilde{O}(m/\delta)$ edges.
Our context is different because we do not have a lower bound on the minimum degree $\delta$ and we want to preserve cuts whose values is not necessarily close to the minimum; and for this reason we can only preserve cuts that are friendly rather than merely non-trivial. 

Before proceeding to discuss the limits of friendly cuts sparsifiers in Section~\ref{subsec:limits} let us motivate them further by presenting our main application.

%
%

\subsection{Application: Faster \GHT}

The question whether the time bound $\min\{ n^{2+o(1)}, m n^{1/2+o(1)} \}$ can be improved is particularly interesting from the perspective of \emph{fine-grained complexity} when comparing constructing a \GHT to two other basic graph problems : triangle and $4$-cycle detection.
In the regime of moderately dense graphs where $m=n^{1.5}$, it is often conjectured that the quadratic barrier cannot be broken for these problems \cite{YZ97,AbboudW14,DudekG19} (and reporting triangles in subquadratic time is $3$-SUM hard \cite{Pat10,KPP16,BPWZ14}).
An application of our friendly cut sparsifiers is a (conditional) separation: \GHT is subquadratic in this regime.

\begin{theorem}[Faster \GHT]
\label{main:alg}
There is a randomized algorithm that constructs a Gomory-Hu Tree of a simple graph on $n$ nodes and $m$ edges in time $(m+n^{1.9})^{1+o(1)}$.
Assuming an almost-linear time \MF algorithm for undirected graphs with polynomially bounded edge weights, the running time becomes $(m+n^{1.75})^{1+o(1)}$.
\end{theorem}

This theorem is achieved by taking the recent $n^{2+o(1)}$-time algorithms \cite{LPS21_focs,AKT21_focs,Zhang21} and replacing the Nagamochi-Ibaraki $w$-cut sparsifier with our new friendly $w$-cut sparsifier from Theorem~\ref{thm:friendly}.
Notably, expander-decompositions are now used \emph{twice} in the state-of-the-art algorithms: once for sparsification and once for reducing the number of queries.

But why is it enough to preserve the friendly cuts?
It is because there is an alternative, more efficient (and simpler) way to compute the unfriendly minimum $s,t$-cuts.
This is stated in the next algorithm, which relies on the recent \pC procedure \cite{LP20,AKT21_stoc} and works even for weighted graphs.
It solves a single-source version of \GHT, which is known to be essentially as hard as \GHT \cite{AKT20,LP21}.

\begin{theorem}[Single-Source Unfriendly Minimum-Cuts]
\label{thm:nonfriendly}
There is an algorithm that, given an undirected graph $G$ on $n$ nodes and $m$ edges with polynomially bounded weights, and a source node $p \in V$,
outputs for every $v \in V(G)\setminus \{ p\}$ a $p,v$-cut $S_v$ such that:
if there is a minimum $p,v$-cut that is unfriendly then the output $S_v$ is a minimum $p,v$-cut.
The running time is $\tilde{O}(m)$ plus the time of $\poly \log{n}$ calls to a \MF algorithm on $O(n)$-node $O(m)$-edge graphs.
\end{theorem}

This algorithm for unfriendly cuts is given in Section~\ref{sec:unfriendly}.
A technical overview of the issues that arise when plugging the friendly sparsifier and this algorithm into the recent \GHT algorithms is given in Section~\ref{sec:GHalg}, and the actual details for proving Theorem~\ref{main:alg} are in Section~\ref{sec:GHalg}.

Could this approach lead us all the way to the conditionally optimal $m+n^{1.5}$ bound?
At a high-level, the $n^{1.75}$ upper bound follows from the following calculation. Each of the $n/w$ calls to \MF in the recent algorithms \cite{LPS21_focs,AKT21_focs,Zhang21} is performed on a sparsifier with $n \sqrt{w}$ edges, making the time bound $n^2/\sqrt{w}$.
An alternative method (better for small $w$) is to use a $w$-partial tree \cite{BHKP07} on the sparsifier, which has time bound $n w^{1.5}$.
Thus the worst-case of $\min\{n^2/\sqrt{w}, n w^{1.5}\}$ over all possible $w$ is $n^{1.75}$.
So all we need is a friendly $w$-cut sparsifier on $\tilde{O}(n)$ edges, for all $w$, as that would lead to a bound $\min\{n^2/w, n w\}=O(n^{1.5})$.
Is that possible?

\subsection{The Limits of Friendly Sparsification}
\label{subsec:limits}

Unfortunately, it is not hard to see that $n \sqrt{w}$ is also a lower bound for friendly $w$-cut sparsifiers.

\begin{observation}[Optimality of our Bound]
For every $n\geq w\geq 1$, there exists a simple graph on $n$ nodes
such that every friendly $w$-cut sparsifier must have $\Omega(n\sqrt{w})$ edges.
\end{observation}

\begin{proof}
For simplicity, we consider only $w=n$ (the proof generalizes easily).
Take a clique on $n$ nodes and replace each node with a $10\sqrt{n}$-clique, connecting the original $n-1$ edges incident at a node, to arbitrary nodes in its $10\sqrt{n}$-clique, but such that no new node gets more than $\sqrt{n}$ original edges. The new graph has $n^{1.5}$ nodes and all of the $\Omega(n^2)$ edges of the original clique must remain in every friendly $w$-cut sparsifier: Each trivial cut of the original clique corresponds to a cut with $n-1$ edges that is friendly in the new graph, because each node has $10 \sqrt{n}-1$ edges to its side but only $\sqrt{n}$ edges across.
\end{proof}

However, the above proof leaves a bit of hope: the lower bound is because of friendly cuts that are not \emph{minimum $s,t$-cuts} for any pair $s,t$. Indeed, all minimum $s,t$-cuts in this construction are trivial with $\leq 11\sqrt{n}$ edges.
So a friendly \emph{minimum $s,t$-cut} sparsifier could contract the entire graph.

\begin{definition}[Friendly Minimum $s,t$-Cut Sparsifier]
\label{def:friendlymin}
We say that $H$ is a \emph{friendly minimum $s,t$-cut sparsifier} for $G$
if for every pair $s,t$ for which all the minimum $s,t$-cuts in $G$ are friendly,
at least one minimum $s,t$-cut is preserved in $H$.
\end{definition}

For the \GHT application, and perhaps others, this is all we need. 
Perhaps surprisingly, it turns out that the desired $\tilde{O}(n)$ bound can indeed be achieved, albeit we do not know how to achieve it in time that is faster than computing a \GHT.
We prove the following theorem in Section~\ref{sec:existensial},
using a structural analysis of the \GHT when the graph is simple. 

\begin{theorem}[Near-Linear Upper Bound]
\label{thm:existensial}
Every simple graph on $n$ nodes has a friendly minimum $s,t$-cut sparsifier with $O(n \log n)$ edges.
Moreover, it can be computed in linear time from a \GHT of the graph.
\end{theorem}

We have thus reached a computational \emph{equivalence} between \GHT and friendly sparsification: faster algorithms exist for one if and only if they exist for the other.

\begin{theorem}[Computational Equivalence]
\label{thm:equiv}
Let $T(n,m)= \Omega(m+n^{1.5})$. 
An  $O(n \log n)$-edge friendly minimum $s,t$-cut sparsifier of a simple graph can be computed in time $T(n,m)\cdot n^{o(1)}$ if and only if a GHT can be computed in time $T(n,m)\cdot n^{o(1)}$.
\end{theorem}

In other words, the only remaining question in the context of \GHT algorithms for simple graphs (without breaking the cubic barrier for multigraphs) is whether an $\tilde{O}(n)$-edge friendly minimum $s,t$-cut sparsifier can be computed in $(m+n^{1.5})^{1+o(1)}$ time.

\medskip

Finally, we ask whether the $\tilde{O}(n)$ upper bound can be obtained for (non-simple) multigraphs, and discover that there is an $\Omega(n^2)$ lower bound.

\begin{observation}
\label{obs:multigraphs}
For every $n$, there exists a multigraph on $n$ nodes and $O(n^2)$ edges,
for which every friendly minimum $s,t$-cut sparsifier has $\Omega(n^2)$ edges.
\end{observation}
\begin{proof}
Take a cycle on $n$ nodes $v_1,\ldots,v_n$ where the weight of each edge (i.e. the number of parallel edges between the pair) is defined in the following alternating manner.
Set $w(v_i,v_{i+1})=\eps n$ if $i$ is even, and set $w(v_i,v_{i+1})= n$ if $i$ is odd, and finally set $w(v_1,v_n)=\eps n -1$. (The parameter $\eps$ can be chosen based on the friendliness parameter, in our case $\eps=0.4$ suffices.)
The (only) minimum $v_i,v_{i+1}$-cut, when $i$ is even, is the \emph{friendly} cut $(\{v_1,\ldots,v_{i}\},\{v_{i+1},\ldots,v_{n}\})$.
Any sparsifier must preserve all the edges of these $n/2$ cuts, whose total number of edges is $\geq n/2 \cdot \eps n= \Omega(n^2)$.
\end{proof}

Therefore, a different notion of sparsification seems to be required for breaking the cubic barrier for weighted graphs. Perhaps \emph{terminal sparsification}, discussed next.

\subsection{Terminal Sparsification}
\label{subsec:terminal}

The techniques of this paper also lead to new \emph{terminal} minimum $s,t$-cut sparsifiers.

\begin{definition}[Terminal Minimum $s,t$-Cut $w$-Sparsifier]
\label{def:terminal}
We say that $H$ is a \emph{terminal minimum $s,t$-cut $w$-sparsifier} of a graph $G$ and terminal set $T$ if it preserves all cuts with $\leq w$ edges that are a minimum $s,t$-cut for some pair of terminals $s,t \in T$ in $G$. 
\end{definition}

An ideal analogue of Nagamochi-Ibaraki sparsification would be a terminal minimum $s,t$-cut sparsifier with $O(|T|w)$ edges.
It is not hard to see that such a bound is existentially possible, even for multigraphs, and that it can be constructed from a \GHT.\footnote{Take the terminal \GHT, which is a cut-equivalent tree on super-nodes where each super-node contains exactly one terminal. Contract each super-node. Sparsify the resulting $|T|$-node graph with Nagamochi-Ibaraki sparsification to get $O(|T|k)$ edges.}
We show that in almost-linear time we can get a bound that is worse by an additive $+n\sqrt{w}$ term; improving on the $O(nw)$ Nagamochi-Ibaraki bound. The proof is in Section~\ref{sec:terminal}.

%

%
\begin{theorem}
\label{thm:terminals}
There is an $m^{1+o(1)}$ time algorithm that, given a simple graph on $n$ nodes, $m$ edges, and terminal set $T$, computes a terminal minimum $s,t$-cut $w$-sparsifier on $(n \sqrt{w} + |T| w )^{1+o(1)}$ edges.
\end{theorem}

We remark that it would have been possible to achieve our $(n^{1.75}+m)^{1+o(1)}$ upper bound for \GHT using this terminal sparsifier rather than the friendly sparsifier, although in a less elegant way (in our view).
Therefore, we expect the equivalence of Theorem~\ref{thm:equiv} to be extendible to terminal sparsification with $\tilde{O}(|T|w)$ edges as well. 
But can the $O(nw)$ bound be beaten algorithmically for \emph{multigraphs}?
 
%

\subsection{Related Work}
\label{sec:related}

We consider graph sparsification that preserves (certain) cut values. 
This topic has been extremely influential,
and perhaps the first such result is the work of Gomory and Hu~\cite{GH61},
because a \GHT is just a sparsifier for all minimum $st$-cuts in the graph.
This line of research has generated several other influential notions,
including for cut values up to a threshold \cite{NI92},
for directed rooted connectivity \cite{Gabow95}, 
for minimum cuts between subsets of terminals \cite{HKNR98}, 
for global minimum cut \cite{Karger99},
and for all cuts up to some approximation \cite{BeK15}.

More broadly, graph sparsification covers many other useful graph quantities,
including directed reachability~\cite{AGU72},
shortest-path distance \cite{PS89,KNZ14,Gupta01},
resistance distance (i.e., effective resistance) \cite{DKW15},
potential energy of the Laplacian (called spectral sparsification) \cite{ST11}, 
and multicommodity flows (including routings) \cite{Raecke02,Moitra09}.
There are some connections between these quantities,
e.g., spectral sparsification always preserves all the resistances and also all the cuts.

The literature usually makes a distinction between
edge sparsification (which decreases the number of edges by taking a subgraph,
possibly with reweighted edges),
and vertex sparsification (which decreases the number of vertices
by merging or removing vertices, which in some cases produces a minor).
In the latter context, the input graph usually comes with a set of terminals,
whose properties (distances, cuts, etc.) must be preserved. 

Graph sparsification (for cuts and in general) has many downstream applications.
The original motivation for many of the abovementioned results
is to speed up algorithms.
Other uses include reducing storage and/or communication requirement
in streaming and distributed settings,
or to improve the approximation factor (to depend on the number of terminals).
For a more empirical perspective, which addresses a range of objectives
under the names graph summarization and graph coarsening,
see the survey~\cite{LSDK18}. 

\section{Friendly Sparsification}
\label{sec:friendly}

This section proves Theorem~\ref{thm:friendly},
the main sparsification result of this paper. 
%
The main workhorse of our construction is an efficient procedure
that decomposes a graph into (node-disjoint) expanders,
such that the number of edges between these expanders is small.
We thus start with describing the relevant definitions and known algorithms 
in Section~\ref{sec:exp-decomp}.
We then present our deterministic algorithm, which is simpler, 
in Section~\ref{sec:detfriendly}. 
At a high level, its idea is very simple --
compute an expander decomposition and then contract the nodes of each expander. 
The intuition is that it is safe to contract the nodes of an expander,
because it should not be ``split'' by the low-weight cuts we are interested in.
However, the actual procedure must be refined by ``shaving''
some nodes from each expander before contracting it.
Moreover, we need a stronger version of the expander decomposition,
that can handle demands.

Our randomized algorithm,
which is faster and computes a slightly smaller sparsifier,
is presented in Section~\ref{sec:randalg}.
It is based on the same techniques
but uses a more elementary version of expander decomposition (without demands).
This limitation forces us to work in iterations that sparsify the graph gradually,
but the advantage is that this version admits a (randomized) algorithm
that is faster, replacing $n^{o(1)}$ terms with polylogarithmic factors.

\subsection{Preliminaries: Expander Decomposition}
\label{sec:exp-decomp}

We mostly follow notations and definition from~\cite{SW19}.
Let $G=(V,E)$ be an undirected graph with edge capacities. Define the \emph{volume} of $C\subseteq V$ as
$\vol_G(C) := \sum_{v\in C}\deg_G(v)$,
where the subscripts referring to the graph are omitted if clear from the context.
The \emph{conductance} of a cut $S$ in $G$ is $\Phi_G(S) := \frac{\delta(S)}{\min(\vol_G(S),\vol_G(V\setminus S))}$.
The \emph{expansion} of a graph $G$ is $\Phi_{G} := \min_{S\subset V}\Phi_G(S)$. 
If $G$ is a singleton then $\Phi_{G}:=1$ by convention. 
Let $G[S]$ be the subgraph induced by $S\subset V$, and let $G\{S\}$ denote the induced subgraph $G[S]$ but with an added self-loop $e=(v,v)$ for each edge $e'=\{ v,u \}$ where $v\in S,u\notin S$ (where each self-loop contributes $1$ to the degree of a node), so that every node in $S$ has the
same degree as its degree in $G$.
Observe that for all $S\subset V$, $\Phi_{G[S]}\ge\Phi_{G\{S\}}$, because the self-loops increase the volumes but not the values of cuts.
We say that a graph $G$ is a \emph{$\phi$-expander} if $\Phi_{G}\ge\phi$, and we call a partition $V=V_1 \sqcup\cdots\sqcup V_h$ 
a $\phi$-expander-decomposition if $\min_{i}\Phi_{G[V_{i}]}\ge\phi$.

\begin{theorem}[Theorem $1.2$ in~\cite{SW19}]
\label{thm:exp-dec}
Given a graph $G=(V,E)$ of $m$ edges and a parameter $\phi$, one can compute with high probability a partition $V=V_1 \sqcup\cdots\sqcup V_h$ such that $\min_i \Phi_{G[V_{i}]}\ge\phi$
and $\sum_{i}\delta(V_{i})=O(\phi m\log^{3}m)$.
In fact, the algorithm has a stronger guarantee that $\min_i \Phi_{G\{V_{i}\}}\ge\phi$.
The running time of the algorithm is $O(m\log^{4}m/\phi)$.
\end{theorem}

For our deterministic upper bound we will need a deterministic version of Theorem~\ref{thm:exp-dec}, which exists, albeit with worse bounds~\cite{CGLNPS20}.
If one is paying the extra bounds anyway, it is possible to introduce additional power to the decomposition by introducing \emph{demands on the nodes} (to be used instead of the degrees when computing the volume) which greatly simplifies the proof that uses them (as a black box). 
Suppose that we are also given a \emph{demand vector} $\dem\in \mathbb{R}_{\geq 0}^{V}$,
the graph $G=(V,E)$ is a \emph{$(\phi,\dem)$-expander} if for all subsets $S\subseteq V$, 
$\Phi_G^{\dem}(S) := \frac{\delta(S)}{\min(\dem(S),\dem(V\setminus S))} \geq \phi$. 
The following theorem statement from \cite[Theorem $III.8$]{LP20} gives a deterministic algorithm. 
It is actually proved in \cite[Corollary 2.5]{LS21},
by extending techniques from~\cite{CGLNPS20}. 

\begin{theorem}[\cite{LS21}]
\label{thm:exp-dec-dem}
Fix $\varepsilon>0$ and any parameter $\phi>0$.
Given an edge-weighted, undirected graph $G=(V,E,w)$
and a demand vector $\dem \in R_{\geq 0}^{V}$,
there is a deterministic algorithm running in time $O(m^{1+\varepsilon})$
that computes a partition $V=V_1 \sqcup\cdots\sqcup V_h$ such that
\begin{enumerate}
\item 
For each $i\in[h]$, define a demand vector $\dem_i\in \mathbb{R}^{V_i}_{\geq 0}$ given by $\dem_i(v)=\dem(v) + w(E({v}, V\setminus V_i))$ for all $v \in V_i$.
Then, the graph $G[V_i]$ is a $(\phi, \dem_i)$-expander.
\item 
The total weight of inter-cluster edges is
$w(E(V_1, \ldots, V_h)) = \sum_i w(E(V_i,V\setminus V_i)) \leq B\cdot \phi \dem(V)$
for $B = (\log n)^{O(1/\varepsilon^4)}$.
\end{enumerate}
\end{theorem}

\subsection{Deterministic Algorithm}
\label{sec:detfriendly}

\begin{proof}[Proof of Theorem~\ref{thm:friendly} (deterministic algorithm)]
Given a simple graph $G=(V,E)$ and a parameter $w$,
the algorithm first computes an expander decomposition of $G$ into expanders $H_1,\ldots,H_{\ell}$,
using Theorem~\ref{thm:exp-dec-dem} with parameter
$\phi=2^{-\log^{1/2}{n}}=n^{-o(1)}$ 
and demand function $\dem(v)=\phi^{-1}\sqrt{w}$ for all $v \in V$.
By setting the parameter $\epsilon=(\log n)^{-1/9}=o(1)$,
the running time is $m^{1+\epsilon}$
and the outer-edges depend on
$B=(\log n)^{O(1/\epsilon^4)} \leq O(\phi^{-1}/\log n) \leq n^{o(1)}$.
Then each expander $H_i$ is a $(\phi, \dem_i)$-expander
where $\dem_i(v)=\dem(v) + |E({v}, V\setminus H_i)|$ for all $v \in H_i$.
And since the total demand is $\dem(V) = n\cdot \phi^{-1}\sqrt{w}$,
the total number of outer-edges is
\[
  m_0
  := \sum_{i} |E(H_i, V \setminus H_i)|
  \leq B\cdot \phi \dem(V)
  =    B n\sqrt{w}
  \leq n^{1+o(1)} \sqrt{w}. 
\]

Second, the algorithm computes for each expander $H_i$ its \emph{shaved expander $H_i'$},
obtained by removing from $H_i$ (simultaneously) all nodes $v \in H_i$
that satisfy at least one of these two conditions:
\begin{itemize} \compactify
\item it has a low degree $\deg_G(v)< 10\sqrt{w}$; or 
\item more than $10\%$ of its degree goes outside of the expander,
  i.e., $|E( \{v\} , V\setminus H_i )| > 0.1 \deg_G(v)$.
\end{itemize}
Finally, the algorithm contracts every shaved expander $H_i'$,
and return the resulting contracted graph $G'$.

The running time is dominated by $m^{1+o(1)}$ the complexity of the expander decomposition procedure; the other operations take linear time.
The following two claims conclude the proof.

\begin{claim}
\label{claim:monochromatic}
Let $S \subseteq V$ be a friendly cut in $G$ that has weight $\delta(S) \leq w$.
Then $G'$ preserves this cut $S$ (i.e., never contracts two nodes that are on different sides).
\end{claim}

\begin{proof}
Assume for contradiction that two nodes $x \in S, y \notin S$
are contracted into the same node in $G'$,
i.e., they belong to the same shaved expander $H_i'$.
Consider the projection of the cut $S$ on the expander $H_i$,
given by $L := H_i \cap S$ and $R := H_i \cap (V \setminus S) = H_i \setminus L$,
which are both non-empty because $x \in L$ and $y \in R$.
Now since $H_i$ is a $(\phi,\dem_i)$-expander, 
we must have
$\frac{|E(L,R)|}{\minn{\dem_i(L),\dem_i(R)} } \geq \phi$.
Clearly, $|E(L,R)| \leq |E(S,V\setminus S)| \leq w$,
and thus
\[
  \minn{\dem(L),\dem(R)}
  \leq \minn{\dem_i(L),\dem_i(R)}
  \leq \frac{|E(L,R)|}{\phi}
  \leq \phi^{-1} w .
\]

We will now show a lower bound on $\dem(L)$.
Since the cut $S$ is friendly,
$|E(\{x\},S)| \geq 0.4 \deg_G(x)$.
Moreover, since $x$ was not shaved,
we know that $\deg_G(x) \geq 10\sqrt{w}$
and that at most $10\%$ of its degree goes outside of the expander,
i.e., $|E(\{x\},V\setminus H_i)| \leq 0.1 \deg_G(x)$.
Combining these three inequalities, 
\[
  |E(\{x\},L)|
  = |E(\{x\},S\cap H_i)|
  \geq (0.4 - 0.1) \deg_G(x)
  \geq 3\sqrt{w}.
\]
Relying on the key property that the graph is simple,
the number of \emph{nodes} in this set $L$ must be $|L| \geq 3\sqrt{w}$.
By the definition of our demand function,
all nodes $v \in V$ have $\dem(v) = \phi^{-1}\sqrt{w}$,
and therefore $\dem(L) = \phi^{-1}\sqrt{w}\cdot |L| \geq 3\phi^{-1} w$.
The argument for $\dem(R)$ is symmetric,
and we arrive at $\minn{\dem(L),\dem(R)} \geq 3\phi^{-1} w$,
in contradiction to the above.
\end{proof}

\begin{claim}
\label{claim:edgebound}
The number of edges in $G'$ is $n^{1+o(1)}\sqrt{w}$.
\end{claim}

\begin{proof}
$G'$ has three kinds of edges, all originating from $G$: 
\begin{enumerate}
\item The outer-edges of the expander decomposition.
  Their number is $m_0 \leq B n\sqrt{w} \leq n^{1+o(1)} \sqrt{w}$. 

\item Edges incident to nodes that were shaved due to having low degree.
  Their number can be upper bounded by $n \cdot 10 \sqrt{w} = O(n\sqrt{w})$.

\item Edges incident to nodes that were shaved due to having more than $10\%$ of their degree going outside their expander.
  We upper bound the number of these edges by $20 m_0$ as follows.
  Let $X \subseteq V_{j-1}$ be the set of these shaved nodes,
  and let $d_{out}(v)$ be the number of edges from $v$ to nodes outside its expander. 
  Observe that $m_0 = 1/2 \cdot \sum_{v \in V} d_{out}(v)$,
  and by definition, every $v\in X$ satisfies
  $d_{out}(v) > 0.1 \deg_G(v)$.
  Altogether,
  \[
    \sum_{v \in X} \deg_G(v)
    \leq \sum_{v \in X} 10 \cdot d_{out}(v)
    \leq 10  \sum_{v \in V} d_{out}(v)
    = 20 m_0.
  \]
\end{enumerate}
Altogether, the total number of edges is
$O(m_0 + n\sqrt{w}) \leq n^{1+o(1)} \sqrt{w}$,
proving the claim.
\end{proof} 

This completes the proof of the deterministic algorithm
in Theorem~\ref{thm:friendly}. 
\end{proof}

The more efficient bound with polylogarithmic factors instead of $n^{o(1)}$, using randomized expander-decomposition algorithms, is given next.

\subsection{Randomized Algorithm}\label{sec:randalg}

Here, we prove the improved bound in Theorem~\ref{thm:friendly} that uses a randomized expander-decomposition algorithm.
The arguments are similar to those in the deterministic algorithm section above, but they are applied in a recursive manner.

\begin{proof}[Proof of Theorem~\ref{thm:friendly} (randomized algorithm)]
Given a simple graph $G=(V,E)$ and a parameter $w$,
the algorithm proceeds in iterations, 
where each iteration $j=0,1,2,\ldots$ produces 
a friendly $w_j$-cut sparsifier $G_j$ of $G$ with $m_j\leq K n\sqrt{w_j}$ edges,
for $w_j = 4^{-j}(m/n)^2$ and $K=\log^{O(1)} n$ to be determined later.
The iterations continue as long as $w_j\ge w$,
hence the last iteration $j$ satisfies $w_j < 4w$
and its sparsifier $G_j$ is reported as the final sparsifier. 
Iteration $0$ produces $G_0=G$ itself which is trivially a sparsifier.

Each iteration $j\ge 1$ uses the sparsifier $G_{j-1}=(V_{j-1},E_{j-1})$
produced in the previous iteration, as follows.
The first step (of iteration $j$) uses Theorem~\ref{thm:exp-dec}
with parameter $\phi=0.01/B$,
where $B=O(\log^3 n)$ denotes the factor in its out-edges gurantee.
It is used to decompose the graph $G_{j-1}$ into expanders $H_1,\ldots,H_{\ell}$,
but only after adding $\phi^{-1}\sqrt{w_j}$ self-loops to each node in $G_{j-1}$,
to increase its volume (but not counting them in the degree).

Then each expander $H_i$ is a $\phi$-expander,
and the total number of outer-edges is
\[
  m'_j
  := \sum_{i=1}^\ell |E_{G_{j-1}[V']}(H_i, V' \setminus H_i)|
  \leq B\cdot \phi (m_{j-1} + n \phi^{-1}\sqrt{w_j}). 
\]

The second step (of iteration $j$) is to compute
for each expander $H_i$ its \emph{shaved expander $H_i'$},
obtained by removing from $H_i$ (simultaneously) every node $v \in H_i$
that satisfies at least one of these two conditions:
\begin{itemize} \compactify
\item it has a low degree $\deg_{G_{j-1}}(v) < 10\sqrt{w_j}\cdot\size_G(v)$; or 
\item more than $10\%$ of its degree goes outside of the expander,
  i.e., $|E_{G_{j-1}}( \{v\} , V_{j-1}\setminus H_i )| > 0.1 \deg_{G_{j-1}}(v)$.
\end{itemize}

The third and final step (of iteration $j$) produces the sparsifier $G_j$
by contracting in the sparsifier $G_{j-1}$ every shaved expander $H_i'$,
and removing self-loops (but keeping parallel edges). 
The resulting $G_j$ is a contracted graph of $G$,
because a node in a shaved expander is itself a contraction of nodes from $V$.

The running time of each iteration $j$ is dominated by the complexity of the expander decomposition procedure,
which is $O(m_{j-1}\cdot B \phi^{-1}\log n) = O(m_{j-1} B^2\log n)$;
the other operations take linear time.
The number of iterations is $O(\log \frac{(m/n)^2}{w})=O(\log n)$,
and in fact $\sum_j m_{j-1} = O(m)$,
and therefore the overall running time is $\tO(m)$.

The correctness is proved in the following two claims.

\begin{claim}
Let $S \subseteq V$ be a friendly cut in $G$ that has weight $\delta_G(S) \leq w$.
Then every $G_j$ preserves this cut $S$ (i.e., never contracts two nodes that are on different sides).
\end{claim}

\begin{proof}
We prove this by induction on $j$. 
The case $j=0$ holds trivially, because $G_0=G$.
Consider then $j\geq 1$, and observe that $\delta_G(S) \leq w \le w_j$.
By the induction hypothesis, $S$ is preserved in $G_{j-1}$;
hence, every node in $V_{j-1}$ is the contraction of nodes
either only from $S$ or nodes only from $V\setminus S$.
With a slight abuse of notation, we can thus think of $S$ as a subset of $V_{j-1}$,
and thus also a cut in $G_{j-1}$ with $\delta_G(S)\leq w$ edges. 
%
Assume for contradiction that two nodes $x \in S, y \notin S$
are contracted into the same node in $G_j$.
This must happen because some shaved expander $H_i'$ contracts them together,
more precisely, $x$ is in a contracted node $\bar x\in H_i'$,
and $y$ is in a contracted node $\bar y\in H_i'$.
Consider now the projection of the cut $S$ on the expander $H_i$, given by
\[
  L := H_i \cap S
  \quad\text{ and }\quad
  R := H_i \setminus S = H_i \setminus L,
\]
which are both non-empty because $\bar x \in L$ and $\bar y \in R$.
Now since $H_i$ is a $\phi$-expander, 
we must have
$\frac{|E_{G_{j-1}}(L,R)|}{\minn{\vol(L),\vol(R)} } \geq \phi$.
Clearly, $|E_{G_{j-1}}(L,R)| \leq |\delta_G(S)| \leq w$,
and thus
\[
  \minn{\vol(L),\vol(R)}
  \leq \frac{|E_{G_{j-1}}(L,R)|}{\phi} 
  \leq \phi^{-1} w .
\]

We now show a lower bound on $\vol(L)$,
and a symmetric argument applies to $\vol(R)$.
Since the cut $S$ is friendly,
$|E_G(\{x\},S)| \geq 0.4 \deg_G(x)$.
This inequality in fact holds for every node $x'\in V$ in the same contracted node $\bar x$, and summing all these inequalities we get
\[
  |E_{G_{j-1}}(\{\bar x\},S)|
  = \sum_{x'\in \bar x} |E_G(\{x'\},S)| 
  \geq \sum_{x'\in \bar x} 0.4 \deg_G(x')
  = 0.4 \deg_{G_{j-1}}(\bar x). 
\]
In addition, since $\bar x$ was not shaved, 
we know that at most $10\%$ of its degree goes outside of the expander,
i.e.,
\[
  |E_{G_{j-1}}( \{\bar x\}, V_{j-1}\setminus H_i )| \leq 0.1 \deg_{G_{j-1}}(\bar x),
\]
and that its degree is
\[
  \deg_{G_{j-1}}(\bar x) \geq 10\sqrt{w_j} \cdot\size_G(\bar x).
\]
Combining the last three inequalities and recalling that $L=S\cap H_i$, 
\[
  |E_{G_{j-1}}(\{\bar x\},L)|
  \geq (0.4 - 0.1) \deg_{G_{j-1}}(x)
  \geq 3\sqrt{w_j} \cdot\size_G(\bar x).
\]
Relying on the key property that $G_{j-1}$ is obtained from a simple $G$
(by contractions), we also have
$|E_{G_{j-1}}(\{\bar y\},L)| \leq \size_G(\bar x)\cdot |L| $,
and altogether
$|L| \geq 3\sqrt{w_j}$. 
Moreover, all nodes $v \in L$ have, because of the added self-loops, 
$\vol(v) \geq \phi^{-1}\sqrt{w_j}$, 
and therefore $\vol(L) \geq |L|\cdot \phi^{-1}\sqrt{w_j} \geq 3\phi^{-1} w_j$.
By a symmetric argument for $\vol(R)$,
we arrive at the contradiction that
$\minn{\vol(L),\vol(R)} \geq 3\phi^{-1} w$. 
\end{proof}

\begin{claim}
The number of edges in $G_j$ is at most $K n\sqrt{w_j}$.
\end{claim}

\begin{proof}
Each sparsifier $G_j$ has three kinds of edges, all originating from $G_{j-1}$: 
\begin{enumerate}
\item The outer-edges of the expander decomposition.
  Their number is $m'_j \leq B\cdot \phi (m_{j-1} + n \phi^{-1}\sqrt{w_j})$. 

\item Edges incident to nodes that are shaved due to having low degree.
  Their number can be upper bounded by
  $\sum_{v\in V_{j-1}} 10\sqrt{w_j}\cdot\size_G(\bar y)
  \leq 10 n\sqrt{w_j}$.

\item Edges incident to nodes that are shaved since more than $10\%$ of their degree goes outside their expander.
  We upper bound the number of these edges by $20 m'_j$ as follows.
  Let $X \subseteq V_{j-1}$ be the set of these shaved nodes,
  and let $d_{out}(v)$ be the number of edges (in $G_{j-1}$) from $v\in V_{j-1}$ to nodes outside its expander. 
  Observe that $m'_j = 1/2 \cdot \sum_{v \in V_{j-1}} d_{out}(v)$,
  and by definition every $v\in X$ satisfies
  $d_{out}(v) > 0.1 \deg_{G_{j-1}}(v)$.
  Altogether,
  \[
    \sum_{v \in X} \deg_{G_{j-1}}(v)
    \leq \sum_{v \in X} 10 \cdot d_{out}(v)
    \leq 10  \sum_{v \in V} d_{out}(v)
    = 20 m'_j.
  \]
\end{enumerate}
Altogether, the total number of edges in $G_j$ is
\[
  m_j 
  \leq 21 m'_j + 10 n\sqrt{w_j}
  \leq 21 B \phi m_{j-1} + (21 B + 10) n\sqrt{w_j} .
\]
Using the induction hypothesis and plugging our parameters choice
$\phi=0.01/B$ and $K = 100B$, we get
$ B \phi m_{j-1} \leq 0.01 K n\sqrt{w_{j-1}} = 0.02 K n\sqrt{w_{j-1}}$
and $ (21 B + 10) \leq 0.5K$,
and the claim follows. 
\end{proof} 

This completes the proof of the randomized algorithm
in Theorem~\ref{thm:friendly}, 
because the for the final $G_j$ is a friendly $w_j$-cut sparsifier
for $w\le w_j< 4w$.
\end{proof}

\section{An $O(n \log n)$-Size Sparsifier for Friendly Minimum Cuts}
\label{sec:existensial}

In this section we show that in order to preserve friendly \textit{minimum} cuts, the total number of edges in the most succinct sparsifier is always at most $\tO(n)$, proving Theorem~\ref{thm:existensial} and the reduction from sparsification to \GHT in the equivalence of Theorem~\ref{thm:equiv}.
Unlike the upper bounds of the previous section, we do not know of a fast algorithm for computing these more efficient sparsifiers, unless we are given a \GHT of the graph.
Recall that a \emph{friendly minimum $s,t$-cut sparsifier}  (Definition~\ref{def:friendlymin} in Section~\ref{sec:intro}) is only required to preserve a minimum $s,t$-cut for pairs that do not have any unfriendly minimum $s,t$-cut.

\begin{proof}[Proof of Theorem~\ref{thm:existensial}]
Let $T$ be a Gomory-Hu tree of $G$. 
Call an edge of $T$ that correspond to an unfriendly minimum cut an \emph{unfriendly edge}.
Define the sparsifier $H$ to be $G$ after contracting the nodes of each connected component in $T$ comprised only of unfriendly edges. 
The contracted graph $H$ satisfies the requirement for a friendly minimum $s,t$-cut sparsifier because for any pair $u,v$, such that the only minimum $u,v$-cut is friendly, at least one minimum $u,v$-cut (the one in $T$) has survived the contractions.

Next, we show that $|E(H)|\leq O(n\log n)$; observe that $|E(H)|$ is at most the total weight of friendly edges in $T$. Indeed, each edge in $H$ (where parallel edges are counted separately) contributes to the weight of at least one friendly edge in the Gomory-Hu tree.
Root $T$ by an arbitrary node $r$. We first show that the total weight of friendly edges on any path in $T$ from a node to an ancestor of it on the tree is bounded, up to constant factors, by the total number of nodes of the tree that descend from this path (that is, below it in the rooted tree). 
This lets us ``charge'' a path of weight $w$ to $\Omega(w)$ descendants.
Then, we will decompose $T$ into paths (via a light/heavy decomposition) such that a node is a descendant of $O(\log n)$ paths, giving an upper bound of $O(n \log n)$ on the total weight.

Let $P$ be any path in $T$, and let $e_1$ be a friendly edge on this path, and consider the subpath of $P$ that contains the next $100$ friendly edges $e_2,\ldots,e_{101}$ below $e_1$. 
Our goal is to bound $w(e_1)$, up to constant factors, by the total number of nodes that are \emph{private} descendants of this subpath, in a sense that will become clear.
Suppose that there indeed are $100$ friendly edges in $P$ below $e_1$; otherwise, we will handle $e_1$ by a different, simpler argument.
In between each pair of friendly edges $e_i,e_{i+1}$ there could be a subpath of unfriendly edges. 
Denote the endpoints of the friendly edges by $e_1=(v_1',v_1), \ldots, e_{101}=(v_{101}',v_{101})$ and note that $v_{i}$ is equal to $v_{i+1}'$ if and only if there are no unfriendly edges between $e_i$ and $e_{i+1}$.
Define the private set of node $v_i$, denoted $\private(v_i)$, to be the descendants of all nodes on the subpath of $P$ that are between $v_{i}$ and $v_{i+1}'$, except for the descendants of $v_{i+1}$.\footnote{We assume that a node is a descendant of itself.}
This partitions the descendants of our subpath so that each descendant is assigned to one of the nodes $v_i$.
The key claim is the following; it is reminiscent of \cite[Claim 3.3]{AKT21_focs} but uses the friendliness of cuts rather than their $w$-largeness and non-easiness.

\begin{claim}
There exists an $1 \leq i \leq 100$ such that $\private(v_i) \geq w(e_1)/1000$.
\end{claim}

\begin{proof}
Denote $w := w(e_1)$ and suppose for contradiction that $\private(v_i)<w/1000$ for all $i \in [100]$.
Since $e_1$ is friendly, we know that $v_1$ must send $>0.4 \deg(v_i) \geq 0.4w$ edges to its descendants in the tree.
By our assumption, the number of descendants of the nodes between $v_1$ and $v_{101}'$ is $< 100 \cdot w/1000 = w/10$, and since our graph $G$ is simple, only $w/10$ edges can go to these nodes.
Therefore, there are $>0.3w$ edges between $v_1$ and the descendants of $v_{101}$, which implies that $w(e_i) > 0.3w$ for all $i\in[101]$.
The same argument above can be repeated for each $e_i$ where $i \in [100]$ since they are all friendly edges, and we get that for all $i \in [100]$ there are at least $0.3w\cdot 0.4-w/100=0.11w$ edges between the node $v_i$ and the descendants of $v_{101}$.
Consequently, $w(e_{101}) \geq 100 \cdot 0.11 w = 11w$.
Now, looking up instead of down, and using the fact that $e_{11}$ is also friendly, we conclude that there are at least $0.4 \cdot 11 \cdot w=4.4w$ edges between $v_{101}'$ and all nodes above $v_{101}'$ in the tree. 
But since there are only $< w/100$ nodes above $v_{101}'$ that are not also above $v_1$, we conclude that there are at least $4.4w-0.01w>w$ edges between $v_{101}'$ and the nodes above $v_1$, meaning that $w(e_1)$ must be $>w$, a contradiction.
\end{proof}

Therefore, if each edge $e_1$ charges its weight to the private nodes on this subpath (letting each node pay $\leq 1000$ times), the total charge received by any node (over all edges of $P$) is at most $100 \cdot 1000=O(1)$; because each node is in the subpath of at most $100$ friendly edges.

For the corner case where there are only $<100$ friendly edges in $P$ below $e_1=(v_1',v_1)$ we can simply charge $w(e_1)$ to all descendants of $v_1$.
Since $e_1$ is friendly, there are at least $0.4 w(e_1)$ edges between $v_1$ and its descendants, and due to the simplicity of $G$, the number of descendants must also be $\geq 0.4 w(e_1)$; therefore, each node receives a charge of $\leq 1/0.4$.
And since a node may only be charged $<100 $ times in this way, the total additional charge for each node is $O(1)$.

Finally, we use the heavy-light path-decomposition of $T$ (see~\cite{SleatorT83}). In this decomposition of the rooted tree $T$, each node $u_h$ picks the edge $u_h u_l$ neighboring it from below if $|V(T_{u_l})|>|V(T_{u_h})|/2$, and makes it a \textit{heavy} edge (the rest are light edges).
As explained above, the weights on each path of heavy edges can be bounded by the total amount of nodes descending from it. Any leaf-root path can contain at most $\log n$ light edges, and so also at most $\log n$ heavy paths. 
Finally, each light edge is treated as its own path, and so each node can be charged at most $2\log n$ times (number of heavy paths and light edges in the path above it to the root), concluding the proof.
\end{proof}

To obtain the reduction of Theorem~\ref{thm:equiv} observe that the sparsifier $H$ in the proof can be obtained from the \GHT $T$ in $O(m)$ time.

\section{Single-Source Unfriendly Cuts in $\tilde{O}(m)$ Time}
\label{sec:unfriendly}

The main result of this section is an algorithm for computing all minimum $p,v$-cuts for a single source $p$ that are \emph{unfriendly}.
The main tool that we use is the \pC procedure discovered independently by Li and Panigrahi \cite{LP20} (for global minimum cut algorithm in weighted graphs) and by Abboud, Krauthgamer, and Trabelsi \cite{AKT21_stoc} (for the first subcubic \GHT algorithm for simple graphs), and within a short time span has found several interesting applications \cite{LP21,CQ21,MN21,LNPSS21,LPS21_focs,AKT21_focs,Zhang21}.
In particular, it was used by Li and Panigrahi \cite{LP21} along with randomized sampling to solve the $(1+\eps)$-approximate single-source minimum cuts problem using $\tilde{O}(1)$ queries to \MF.

\begin{definition}[Minimum Isolating Cuts]
Consider a weighted, undirected graph $G=(V,E)$ and a subset $R \subseteq V$ ($|R|\geq 2$).
The minimum isolating cuts for $R$ is a collection of sets $\{ S_v: v \in R \}$ such that for each node $v \in R$, the set $S_v$ is the side containing $v$ of a minimum cut separating $\{v \}$ and $R \setminus \{v\}$, i.e. for any set $S$ satisfying $v \in S$ and $S \cap (R \setminus \{v \}) = \emptyset$, we have $\delta (S_v) \leq \delta(S)$. 
\end{definition}

\begin{lemma}[The \pC Procedure (\cite{LP20,AKT21_stoc})]
\label{lem:proc_main}

There is an algorithm that, given an undirected graph $G=(V,E,c)$ on $n$ nodes and $m$ edges of total weight $c(E)$ and a subset $R \subseteq V$, computes the minimum isolating cuts $\{ S_v : v \in R \}$ for $R$ in $G$ using $O(\log |R|)$ calls to \MF on  graphs with $O(n)$ nodes, $O(m)$ edges, and $O(c(E))$ total weight, and takes $\tilde{O}(m)$ deterministic time outside of these calls.
\end{lemma}

\begin{theorem}[Approximate Single-Source Min-Cuts (\cite{LP21})]
\label{thm:apx}
Let $G$ be an undirected graph on $n$ nodes and $m$ edges, with polynomially bounded weights, and let $p \in V$.
There is an algorithm that outputs, for each node $v \in V \setminus \{p\}$, a $(1+\eps)$-approximation of $\MC(p,v)$, and runs in $\tilde{O}(m)$ time plus $\poly\log{n}$ calls to \MF on $O(n)$-nodes, $O(m)$-edge graphs.
\end{theorem}

The following two similar lemmas about unfriendly cuts are the keys for utilizing the \pC procedure (with the help of an approximation algorithm) to compute the exact minimum cuts.
The idea is that if $v$'s minimum $p,v$-cut is unfriendly, e.g. due to $v$ having too many edges to $p$'s side (observe that no other node $x$ in $v$'s side could have more than half of its edges to $p$'s or else the minimum $p,v$-cut is better off by putting $x$ on the other side), then all other nodes $u$ in the cut must have a minimum $p,u$-cut whose value is smaller than $v$'s cut by a constant factor. 
This happens because moving $v$ to the other side gives a $p,u$-cut of much smaller value.
Therefore, a $(1+\eps)$-approximation to the maximum-flow values to $p$ allows us to distinguish between $v$ and its cut members $u$, and thus to only pick $v$ as terminal out of the nodes in the cut when using the \pC procedure, which isolates $v$.

\begin{lemma}
\label{unfriendly1}
If the minimum $p,v$-cut $S, v \in S$ satisfies $|E(\{v\}, V \setminus S)| > 0.6 \cdot \deg(v)$, then for all $v' \in S \setminus \{s\}$ $\lambda_{p,v'} \leq 0.8 \cdot \lambda_{p,v}$.
\end{lemma}

\begin{proof}
Our goal is to show that the cut $S \setminus \{v\}$ has value $\leq 0.8 \lambda_{p,v}$, which establishes the upper bound on $\lambda_{p,u}$ for all $u \in S \setminus \{v\}$.
First, notice that $\lambda_{p,v} = \delta(S) \leq \deg(v)$ because the degree is always an upper bound on the minimum cut.
Since $\delta(S)= |E(\{v\}, V \setminus S)| + |E(S \setminus \{v \}, V \setminus S)|$ and $|E(\{v\}, V \setminus S)| > 0.6\cdot \deg(v)$ we know that $ |E(S \setminus \{v \}, V \setminus S)| < \lambda_{p,v} - 0.4\cdot \deg(v)$.
Moreover, $| E(S \setminus \{v \}, \{v\}) | < 0.4 \cdot \deg(v)$ because $>0.6 \deg(v)$ of $v$'s edges leave $S$.
Thus, $\delta( S \setminus \{v\}) = | E(S \setminus \{v \}, \{v\}) | + |E(S \setminus \{v \}, V \setminus S)| <  \lambda_{p,v} - 0.6\cdot \deg(v) + 0.4 \cdot \deg(v) < \lambda_{p,v} - 0.2 \cdot  \deg(v) \leq 0.8 \cdot \lambda_{p,v}$.

\end{proof}

\begin{lemma}
\label{unfriendly2}
If the minimum $p,v$-cut $S, v\in S$ satisfies $|E(\{p\}, S)| > 0.6 \cdot \deg(p)$, then for all $v' \in (V \setminus S) \setminus \{p\}$ $\lambda_{p,v'} \leq 0.8 \cdot \lambda_{p,v}$.
\end{lemma}
\begin{proof}
The goal is to show that the cut $(V \setminus S) \setminus \{p \}$ has value $\leq 0.8 \lambda_{p,v}$; the proof is similar to the one above.
First, $\lambda_{p,v} \leq \deg(p)$.
Then, since $\delta(S)= |E(\{p\}, S)| + |E((V \setminus S) \setminus \{p \}, S)|$ and $|E(\{p\}, S)|>  0.6 \cdot \deg(p)$ we know that $|E((V \setminus S) \setminus \{p \}, S)| < \lambda_{p,v} - 0.6\cdot \deg(p)$.
Moreover, $ |E( \{ p \} , (V \setminus S) \setminus \{p \}) <0.4 \cdot \deg(p)$, implying that 
$\delta( (V \setminus S) \setminus \{p \}  ) = 
|E( \{ p \} , (V \setminus S) \setminus \{p \})  + |E((V \setminus S) \setminus \{p \}, S)|
 < 0.4 \cdot \deg(p) +  \lambda_{p,v} - 0.6\cdot \deg(p)  < \lambda_{p,v} - 0.2 \cdot  \deg(p) \leq 0.8 \cdot \lambda_{p,v}$.
\end{proof}

%

We are ready to prove the main result of this section.
 
\begin{proof}[Proof of Theorem~\ref{thm:nonfriendly}]
The algorithm is as follows. Let $\eps=\delta=0.01$.
\begin{enumerate}
\item Use the algorithm of Theorem~\ref{thm:apx} to get an estimate $c'(v)$ for all nodes $v \in V\setminus \{p\}$ such that $\lambda_{p,v} \leq c'(v) \leq (1+\eps) \cdot \lambda_{p,v}$.
\item For each $i\in \{0,\ldots, \log_{(1+\delta)}n \}$ compute the set $T_i = \{ v \in V \mid c'(v) \geq (1+\delta)^{i} \} \cup \{ p\}$ of nodes whose estimate is at least $(1+\delta)^{i}$.
\item For each of the $O(\log{n})$ sets $T_i$ make a call to the \pC procedure of Lemma~\ref{lem:proc_main} on $G$ where the set of terminals is $T_i$. For each $v \in T_i$ let $S_v$ be the returned isolating cut of $v$.
Update the estimate $c'(v)$ to be the minimum between $c'(v)$, the value of $S_v$, and the value of $S_p$ (the returned isolating cut for $p$).
\item Return the estimates of all nodes.
\end{enumerate}

The running time of the algorithm is the time of the approximate single-source algorithm plus the time for $O(\log{n})$ calls to the \pC procedure.

For the correctness, let $v$ be any node $v\in V$ such that the minimum $v,p$-cut $C_v$ is unfriendly and it has value $w:=\lambda_{p,v}$. We will show that $c'(v)=w$ by the end of the algorithm. Let $i$ be such that $(1+\delta)^i \leq w \leq (1+\delta)^{i+1}$. There are two cases:
\begin{itemize}

\item $C_v$ is unfriendly because $|E( \{v \}, V\setminus C_v )| > 0.6 \cdot \deg(v)$: 
In this case, Lemma~\ref{unfriendly1} implies that for all $u \in C_v \setminus \{v\}$ we have $\lambda_{p,u} <0.8 \cdot w$ and therefore $c'(u)< 0.8 \cdot (1+\eps) \cdot w = 0.9 w < w/(1+\delta) \leq (1+\delta)^i$, implying that $u \notin T_i$.
Thus, $T_i \cap C_v = \{ v \}$ and the minimum isolating cut for $v$ with respect to $T_i$ has value at most $\delta(C_v)=w$, which means that the \pC procedure will return a minimum $p,v$-cut (the cut $S_v$).

\item $C_v$ is unfriendly because $|E( \{p \}, C_v )| > 0.6 \cdot \deg(p)$:
 In this case, Lemma~\ref{unfriendly2} implies that for all $u \in (V \setminus C_v) \setminus \{p\}$ we have $\lambda_{p,u} <0.8 \cdot w$ and therefore $c'(u)<(1+\delta)^i$ and $u \notin T_i$.
 Thus, $(V \setminus C_v) \cap T_i = \{ p \}$ and thus the minimum isolating cut for $p$ with respect to $T_i$ has value at most $\delta(C_v)=w$, which means that the \pC procedure will return a minimum $p,v$-cut (the cut $S_p$).
 \end{itemize}

\end{proof}

\section{The Gomory-Hu Tree Algorithm}
\label{sec:GHalg}

In this section we explain how to embed our new friendly cut sparsifiers (and the complementary algorithm for unfriendly cuts) into the recent algorithms for \GHT.
We first present an overview of our methods, following by the technical details required to prove Theorem~\ref{main:alg} and the reduction from \GHT to friendly minimum cut sparsifiers in Theorem~\ref{thm:equiv}.


\subsection{Technical Overview.}

It is likely that each of the recent $n^{2+o(1)}$ algorithms for \GHT \cite{LPS21_focs,AKT21_focs,Zhang21} can be sped up by using our sparsifier instead of Nagamochi-Ibaraki. 
However, as we attempt to explain next, due to the subtle combination of several ingredients in each of these algorithms, it has to be done in just the right way.
The order in which the different ingredients is invoked is different in each of the three algorithms; we have found the presentation of Abboud, Krauthgamer, and Trabelsi \cite{AKT21_focs} to be the easiest to modify.
Let us now explain the issues that arise in this process.

All of these algorithms follow the paradigm of the Gomory-Hu algorithm where there is an intermediate tree $T$ that gets iteratively refined until it becomes the full \GHT.
A refinement step splits one of the super-nodes $V_i$ of $T$ into two or more super-nodes, using minimum $s,t$-cuts that are computed in the \emph{auxiliary graph} $G_i$ of the super-node $V_i$.
The most expensive (and interesting) part of all three algorithms is to compute single-source minimum cuts from a pivot $p \in V_i$ to the other nodes in $V_i$.
Importantly, even though there are $\Omega(n)$ refinement steps throughout the algorithm, the total size of all auxiliary can be upper bounded by $O(m \log n)$.
Therefore, when refining a super-node on $n_i$ nodes and $m_i$ edges, it is important to only spend time proportional to $n_i$ and $m_i$, not $n$ and $m$.

Our friendly sparsification result of Theorem~\ref{thm:friendly} and the algorithm of Theorem~\ref{thm:nonfriendly} are easy to adapt when solving the single-source problem in $(m+n^{1.75})^{1+o(1)}$ time (assuming linear-time \MF).
However, the sparsification result only works for \emph{simple} graphs, and although our input graph is simple, its auxiliary graphs are not.
This is the source of all troubles.
We need to be able to solve the single-source problem in $G_i$ while only paying  $(m_i+n_i^{1.75})^{1+o(1)}$ time, not $(m+n^{1.75})^{1+o(1)}$, but this seems impossible because we can only sparsify the auxiliary graph directly down to $n \sqrt{w}$ edges, not $n_i \sqrt{w}$.
This was not an issue with Nagamochi-Ibaraki sparsification because it can sparsify the auxiliary graph down to $O(n_i w)$ edges.

One approach for resolving this issue (that we were unable to make work) is to first compute a friendly $w$-cut sparsifier $H_w$, for all $w \in \{2^i\}_{i=0}^{\log n}$, then compute the \GHT of each $H_w$, and then use them to answer queries efficiently when computing the \GHT of $G$.
However this approach is not efficient enough because it is unclear how to compute the \GHT of $H_w$ faster than using the $(m \sqrt{n})^{1+o(1)}$ bound for \GHT \cite{LPS21_focs} which gives $n\sqrt{w} \cdot \sqrt{n} = \Omega(n^2)$.
We need to be able to sparsify each auxiliary graph inside the \GHT algorithm.

The approach that does work for us is to sparsify each auxiliary graph \emph{using the same original sparsifier of the input graph}.
We first compute a friendly $w$-cut sparsifier $H_w$ of the simple graph $G$, for all $w \in \{2^i\}_{i=0}^{\log n}$, and then whenever we have an auxiliary graph $G_i$ we compute the ``projection'' of each sparsifier on the auxiliary graph.
The resulting \emph{sparsified auxiliary graphs} (defined in Section~\ref{sec:GHprelims}) may have $m_{w,i} > n_i \sqrt{w}$ edges each, but we can prove that the total number of edges in all of these sparsified auxiliary graphs is upper bounded (up to log factors) by the number of edges in the sparsifier $\sum_i m_{w,i} = n \sqrt{w}$, for each $w$.

%

\subsection{Preliminaries}
\label{sec:GHprelims}
Since the result of this section is obtained by modifying specific theorems in \cite{AKT21_focs}, we will follow their notation and framework.
The reader is referred to \cite[Section 2.1]{AKT21_focs} for the necessary preliminaries on the Gomory-Hu tree algorithmic paradigm, intermediate trees, partition trees, GH-equivalent partition trees, auxiliary graphs, and $k$-partial trees.
On top of these definitions, we will use the following notion of sparsified capacitated auxiliary graph that will be used in our modified analysis (similar to previous used notions~\cite[Section 3.2]{AKT20_b} and~\cite[Section 3]{LPS21_focs}).

\paragraph{Capacitated Auxiliary Graphs}
Recall that an auxiliary graph (in the Gomory-Hu algorithm) of a super-node $V_i$ in a \GHEPT $T$ of a graph is obtained by contracting all connected components in $T \setminus V_i$. 
 A \CAG is constructed the same way as an auxiliary graph, but with parallel edges merged into a single edge that is weighted by their total capacity.

\paragraph{Sparsified \CAGs}
Suppose that $T$ is a \GHEPT of a simple graph $G$ and let $V'$ be one of its super-nodes and $G'$ be the corresponding auxiliary graph.
Let $H$ be a sparsifier of $G$ (that may contract subsets of the nodes).
The sparsified \CAG $H'$ of $G'$ is obtained from $H$ by contracting each connected component of $T \setminus \{V' \}$, i.e. in the same way that $G'$ was obtained from $G$.
Observe that if two nodes $x,y$ of $G$ were contracted in the sparsifier $H$ even though they belong to two different connected components $X$ and $Y$ of $T \setminus \{V' \}$, then they will cause the two components $X$ and $Y$ to be contracted together in the sparsified \CAG $H'$.
The next lemma shows that this (rather dramatic) consistency enforcement makes sense.

\begin{lemma}
Let $H$ be a friendly $w$-sparsifier of $G$ and let $G'$ be an auxiliary graph of $G$ obtained from super-node $V'$ in a \GHEPT $T$ of $G$.
Then the sparsified \CAG $H'$ preserves all friendly minimum $s,t$-cuts of value up to $w$ in $G'$ for any pair $s,t \in V'$.
\end{lemma}

\begin{proof}
Suppose that $(S,T)$ is a minimum $s,t$-cut in $G'$ (and therefore also in $G$) of value $\leq w$, and moreover suppose that it is friendly.
The friendly $w$-sparsifier $H$ must preserve $(S,T)$ and therefore only contracts nodes from within $S$ and from within $T$.
This implies that none of the contractions performed on $H$ in order to obtain the sparsified auxiliary graph $H'$ may contract any node from $S$ with any node from $T$.
Therefore, $(S,T)$ is also preserved in $H'$.
\end{proof}

Next, we bound the total running time in the recursive algorithm. In order to be more general, we phrase it for a \PT. Clearly, any \GHEPT is also a \PT, and so we are left with proving the following lemma.

\begin{lemma}
\label{lem:total_edges_sparsifiers}
Let $H$ be a sparsifier of $G$ and let $T$ be any \GHEPT of $G$ with super-nodes $V_1,\ldots,V_k$ and let $H_i$ be the sparsified \CAG corresponding to $V_i$.
Then, $\sum_{i=1}^k |E(H_i)| = O(|E(H)|)$ and $\sum_{i=1}^k |V(H_i)| = O(|V(G)|)$.\footnote{Just like Lemma $3.12$~\cite{AKT20_b} and Lemma $4.1$~\cite{AKT21_stoc} that are used in the proof, this lemma holds even if $T$ is any \PT that is not necessarily a \GHEPT.}
\end{lemma}

To prove the edges part of Lemma~\ref{lem:total_edges_sparsifiers}, one might first think to apply Lemma $3.12$~\cite{AKT20_b} that is used to bound the total sizes of all \CAGs of a \GHEPT $T$ for an input graph $G$. However, the sparsified \CAGs in our context might not be consistent with the tree, leading to contractions that make this direction problematic. Our strategy in overcoming this issue is to construct a new sparsifier $H'$ from $H$ with precisely the same number of edges, such that the \CAG and the sparsified \CAG of each super-node $V_i\in T$ for $H'$ are the same. This way, we could directly apply Lemma $3.12$.

\begin{proof}[Proof of Lemma~\ref{lem:total_edges_sparsifiers}]
First, in order to bound the total number of nodes in all \CAGs $\sum_{i=1}^{k}V(H_i)$, we sum in a way similar to Lemma $4.1$ in~\cite{AKT21_stoc}. Observe that also here, each sparsified \CAG $H_i$ has at most $|V_i|+\deg_T(V_i)$ nodes, which sums up to $O(|V(G)|)$, as claimed.

Second, for the edge bound, construct a graph $H'$ from $H$ by partitioning each node $u$ of $H$ into the parts intersecting each super-node of $T$, connecting each edge previously connected to $u$ to an arbitrary new part of $u$. Observe that $H'$ has the same number of edges as $H'$ (but maybe a higher number of nodes), and that no node of $H'$ intersects with more than one super-nodes of $T$. By Lemma $3.12$ in~\cite{AKT20_b}, the total number of edges in all \CAGs satisfies $\sum_{i=1}^{k} |E(H'_i)|\leq |E(H')|$ which equals $|E(H)|$.

It remains to bound the total number of edges in each sparsified \CAG $H_i$ of $H$ by its corresponding number in the \CAG $H'_i$ of $H'$. This is correct because each edge in $H_i$ can be identified with an edge in $H'_i$, as follows. Let $e=uw$ be an edge in $H$ that contributes to $H_i$.
Let $u',w'$ be the new parts of $u,w$, respectively, such that the edge $uw$ in $H$ became $u'w'$ in $H'$. The only way the edge $u'w'\in H'$ would not contribute one to the total number of edges in $H'_i$ is if both $u'$ and $w'$ belong to the same subtree in $T$ neighboring $V_i$. However, this would imply that $u,w$ were merged together in $H_i$, in contradiction to the edge $u,w$ showing in $H_i$.
\end{proof}

\subsection{The Single-Source Algorithm}
This is the main algorithmic result that uses all the ingredients in order to solve the single-source problem in an auxiliary graph.
It augments and improves Theorem 4.2 in \cite{AKT21_focs} using the new sparsifiers and by using an additional procedure for the unfriendly cuts.

\begin{theorem}
\label{thm:SSGH}
Suppose there is an algorithm solving \MF on undirected graphs with $m$ edges of polynomially bounded weight in $m^{1+o(1)}$ time.
Given a simple graph $G=(V,E)$ on $N=|V|$ nodes with a designated pivot $p \in V$,
an auxiliary graph $G'$ on the graph nodes $V' \subseteq V(G)$ with $n=|V(G')|, m=|E(G')|$,
a  $2w$-friendly-sparsified \CAG $H_w$ of $G'$ with $n_w$ nodes and $m_w$ edges for each $w \in \{2^i \}_{i=0}^{\log{N}}$,
and a perturbed version $\tilde{G}$ of $G'$ with unique minimum cuts,
one can compute the minimum $(p,v)$-cut in $\tilde{G}$ for all nodes $v \in V'$
in total time
$$m^{1+o(1)} + \sum_{w \in \{2^i \}_{i=0}^{\log{N}}} \min\{ \tilde{O}(m_w w), \frac{N^{1+o(1)}}{w} \cdot m_w \}.$$
Using the existing \MF algorithms, the time bound is 
$$\left( m+n^{1.5}+ \sum_{w \in \{2^i \}_{i=0}^{\log{N}}} \min\left\{ m_w w, \frac{N}{w} \cdot (m_w+n_w^{1.5}) \right\} \right)^{1+o(1)}.$$

\end{theorem}

\begin{proof}

The algorithm that was used to prove Theorem 4.2 in \cite{AKT21_focs} suffices, if we make the following three changes.

\begin{itemize}
\item In the algorithm of \cite{AKT21_focs} there is an estimate $c'(v)$ for each node $v \in V$ that is always an upper bound on the minimum $p,v$-cut value $\lambda_{p,v}$ and each estimate is witnessed by a $p,v$-cut of the same value. Throughout the algorithm, when cuts are found, the estimates decrease until, in the end, they are guaranteed to equal the minimal values.
In \cite{AKT21_focs} the estimates were initialized using the trivial cuts $(\{v\}, V\setminus \{v \})$. 

In this paper, instead, we use the algorithm of Theorem~\ref{thm:nonfriendly} to compute the minimum $p,v$-cuts that are unfriendly for all $v\in V$ and use them as the initial estimates $c'(v)$ and witness cuts $C_v$.
After this initialization, the algorithm will only use \MF calls on friendly sparsifiers but not on $G$ and therefore it will only be guaranteed to find minimum friendly cuts. 
But since whenever we find a cut we only use it if it is better than the estimate, and the optimal cuts that are unfriendly are at hand from the beginning, the rest of the algorithm can operate as if the sparsifiers preserved all cuts of value $\leq 2w$ (as was actually the case when the Nagamochi-Ibaraki sparsifier $G_w$ was used in \cite{AKT21_focs}).
The additional running time for this stage is $m^{1+o(1)}$ assuming an almost-linear time \MF algorithm, and is $\tilde{O}(m+n^{1.5})$ with existing algorithms.

\item The second change is that we use the $2w$-friendly-sparsified auxiliary graph $H_w$ instead of the Nagamochi-Ibaraki sparsifier $G_w$, in every stage $w$ for all $w \in \{2^i \}_{i=0}^{\log{N}}$. 
As discussed above, this does not affect the correctness. 
One subtlety is that the algorithm is also working with a perturbed version $\tilde{G}$ of $G$ (for rather technical reasons\footnote{The purpose was to enforce unique minimum cuts. An alternative approach is to work with minimal (also known as latest) cuts from the source $p$.}) and so, just like in Lemma 2.7 in \cite{AKT21_focs} the sparsifier can have the same perturbation as $G$ by simply adding the same amount $\epsilon(e)$ to each edge $e$ in $H$ that was added to it in $G$.
As a result, the running time bound for each query or each call to the \pC procedure improves from $\tilde{O}(nw)$ to $\tilde{O}(m_w)$, and the upper bound on the time of stage $w$ improves from $N^{1+o(1)}/w \cdot T(n,nw)$ to $N^{1+o(1)}/w \cdot T(n_w,m_w)$, where $T(n,m)$ is \MF time.

\item Finally, in \cite{AKT21_focs} the single-source algorithm assumes that the auxiliary graph has connectivity $\geq \sqrt{N}$ and therefore it can avoid any stage with $w < \sqrt{N}$. This was useful because the existing \MF algorithms have running time $\tilde{O}(m+n^{1.5})$ rather than $m^{1+o(1)}$ and so when they had $m=O(nw)$ from Nagamochi-Ibaraki and $w>\sqrt{n}$ they could assume that the extra $n^{1.5}$ is negligible.
And the way they enforce that the theorem is always used (inside their full \GHT algorithm) with auxiliary graphs with connectivity at least $\sqrt{N}$ is by having a preliminary stage where they compute a $\sqrt{N}$-partial tree \cite{BHKP07}.

Here, we do not assume a lower bound on the connectivity of the auxiliary graph and therefore have to handle stages where $w$ is very small. We handle them by computing a $w$-partial tree \cite{BHKP07} on the sparsifier $H_w$, giving an upper bound of $\tilde{O}(m_w w)$ on the time of stage $w$ which could be better than the $N/w \cdot T(n_w,m_w)$ bound, where $T(n,m)$ is \MF time.

\end{itemize}

To summarize, the correctness of the algorithm is preserved after making these changes, and the time complexity improves. The time for stage $w$ becomes $\min\{ \tilde{O}(m_w w), N^{1+o(1)} T(n_w,m_w)/w \}$ and the time bounds in the theorem are obtained by adding the $m^{1+o(1)}$ cost of the expander decompositions (that are done in exactly the same way as in \cite{AKT21_focs}, i.e. on $G'$) and the $\tilde{O}(T(n,m))$ cost for the unfriendly single-source cuts algorithm.

\end{proof}

\subsection{The Full Gomory-Hu Tree Algorithm} 
We are now ready to prove the main theorem of this paper, Theorem~\ref{main:alg} (as well as the reduction from \GHT to sparsification in Theorem~\ref{thm:equiv}).

The algorithm uses a randomized pivot strategy to compute the \GHT recursively using calls to the single-source algorithm of Theorem~\ref{thm:SSGH}.
It is very similar to the algorithm in Section 4.2 in \cite{AKT21_focs} with the following two differences:

\begin{itemize}
\item The step (Step 1) that computes a $\sqrt{N}$-partial tree is removed, since our single-source algorithm no longer needs the lower bound on the connectivity of an auxiliary graph.
\item Instead, the first step of the algorithm (the new Step 1) computes a $2w$-friendly-sparsifier $H_w$ for each $w \in \{2^i \}_{i=0}^{\log{N}}$ using Theorem~\ref{thm:friendly}.
\item When using the single-source algorithm (in Step 2(c)) we use the algorithm of the new Theorem~\ref{thm:SSGH} rather than Theorem 4.2 of \cite{AKT21_focs}. And thus in Step 2(b) when computing the auxiliary graph $G_i$ and its perturbed version, the algorithm also prepares the $2w$-friendly-sparsified auxiliary graph $H_{w,i}$ of $G_i$ for each $w \in \{2^i \}_{i=0}^{\log{N}}$.
\end{itemize}

The proof of correctness remains unchanged. The running time analysis is a bit more subtle and it uses Lemma~\ref{lem:total_edges_sparsifiers}. 
We still have a recursive algorithm with $O(\log N)$ levels and the goal is to bound the total time of a single level.
Suppose that the \GHEPT that corresponds to a level has $k$ super-nodes $V_1,\ldots,V_k$.
Denote the number of nodes and edges in the auxiliary graph $G_i$ of super-node $V_i$ by $n_i$ and $m_i$, and denote the number of nodes and edges in the $2w$-friendly-sparsified auxiliary graph $H_{w,i}$ of $G_i$ by $n_{w,i}$ and $m_{w,i}$.
Then, by Theorem~\ref{thm:SSGH}, and assuming an almost-linear time \MF algorithm, the total time for the level is:
$$
\sum_{i=1}^k \left( m_i^{1+o(1)} + \sum_{w \in \{2^i \}_{i=0}^{\log{N}}} \min\{ \tilde{O}(m_{w,i} w), \frac{N^{1+o(1)}}{w} \cdot m_{w,i} \} \right).
$$
By Lemma $3.12$ in \cite{AKT20_b} we have that $\sum_{i=1}^k m_i=O(M)$ and by Lemma~\ref{lem:total_edges_sparsifiers} we have that $\sum_{i=1}^k m_{w,i}=O(|E(H_w)|)=\tilde{O}(N\sqrt{w})$. Therefore, the total time of a level can be upper bounded by:
$$
M^{1+o(1)} + \sum_{i=1}^k \sum_{w \in \{2^i \}_{i=0}^{\log{N}}} m_{w,i} \cdot \min\{ w, \frac{N}{w} \} \cdot N^{o(1)}
$$
$$
\leq M^{1+o(1)} + \sum_{w \in \{2^i \}_{i=0}^{\log{N}}} N\sqrt{w} \cdot \min\{ w, \frac{N}{w} \} \cdot N^{o(1)}.
$$
Since $\min\{w^{1.5},N/w^{0.5}\} \leq N^{0.75}$ for all $w \leq N$, with $w=\sqrt{N}$ achieving the maximum, we obtain the $M^{1+o(1)}+N^{1.75+o(1)}$ upper bound.
To prove the reduction from \GHT to fast sparsification in Theorem~\ref{thm:equiv} observe that if the edge bound on the sparsifiers was $\tilde{O}(N)$ instead of $\tilde{O}(N \sqrt{w})$ then the above expression would result in an $M^{1+o(1)}+N^{1.5+o(1)}$ upper bound.
Finally, using the existing \MF algorithms the upper bound on a level becomes:
$$
\sum_{i=1}^k \left( \left( m_i+n_i^{1.5}+ \sum_{w \in \{2^i \}_{i=0}^{\log{N}}} \min\left\{ m_{w,i} w, \ \frac{N}{w} \cdot (m_{w,i}+n_{w,i}^{1.5}) \right\} \right)^{1+o(1)} \right).
$$
By Lemma $4.1$ in \cite{AKT21_stoc} we have that $\sum_{i=1}^k {n_i}=O(N)$ and by Lemma~\ref{lem:total_edges_sparsifiers} we have that $\sum_{i=1}^k {n_{w,i}}=O(N)$, thus we can upper bound the above by: 
$$
  \left( M+N^{1.5}+ \sum_{w \in \{2^i \}_{i=0}^{\log{N}}} \sum_{i=1}^k \min\left\{ m_{w,i} w, \ \frac{N}{w} \cdot (m_{w,i}+n_{w,i}^{1.5}) \right\} \right)^{1+o(1)}
$$
$$
\leq  \left( M+N^{1.5}+ \sum_{w \in \{2^i \}_{i=0}^{\log{N}}} \min\left\{ (\sum_{i=1}^k m_{w,i}) w, \ \frac{N}{w} \cdot \sum_{i=1}^k (m_{w,i}+n_{w,i}^{1.5}) \right\} \right)^{1+o(1)},
$$
because $\sum_{i} \min\{ x_i,y_i\} \leq \min \{ \sum_i {x_i}, \sum_i y_i\}$, which now gives:
$$
\leq  \left( M+N^{1.5}+ \sum_{w \in \{2^i \}_{i=0}^{\log{N}}} \min\left\{ N \sqrt{w} \cdot w, \ \frac{N}{w} \cdot (N \sqrt{w}+N^{1.5}) \right\} \right)^{1+o(1)},
$$
and since $\min\{ N w^{1.5}, N^2/w^{0.5}+N^{2.5}/w \} \leq N^{1.9}$ for all $w\leq N$, with $w=N^{3/5}$ achieving the maximum, the upper bound becomes $M^{1+o(1)}+N^{1.9+o(1)}$.

\section{Terminal Sparsification}
\label{sec:terminal}

In this section we use the proof of Theorem~\ref{thm:friendly} in Section~\ref{sec:detfriendly}, our deterministic algorithm for friendly sparsification using an expander decomposition with node-demands, in order to obtain \emph{terminal sparsifiers} (Definition~\ref{def:terminal} in Section~\ref{subsec:terminal}).
Our goal is to prove Theorem~\ref{thm:terminals}.

%

\begin{proof}[Proof of Theorem~\ref{thm:terminals}]
Let us only explain the differences from the proof of Theorem~\ref{thm:friendly} in Section~\ref{sec:detfriendly}.
The demand function is defined a bit differently:
we still set $\dem(v)=\phi^{-1}\sqrt{w}$ for all non-terminals $v \in V\setminus T$, but the demand of terminals is larger, we set $\dem(v)=3\phi^{-1}w$ for all $v \in T$.
This changes the total demand to $\dem(V) = n\cdot \phi^{-1}\sqrt{w} + 3|T|\cdot \phi^{-1}w$, and the total number of outer-edges to
\[
  m_0
  := \sum_{i} |E(H_i, V \setminus H_i)|
  \leq B\cdot \phi \dem(V)
  =    B n\sqrt{w} + 3B |T| {w}
  \leq n^{o(1)}\cdot (n \sqrt{w} + |T|w). 
\]

This causes the upper bound of Claim~\ref{claim:edgebound} on the number of edges in the sparsifier to change to $n^{o(1)}\cdot (n \sqrt{w} + |T|w)$, simply by plugging in the new value of $m_0$.

The only interesting change is in the correctness proof of Claim~\ref{claim:monochromatic}, which becomes the following.

\begin{claim}
Let $S \subseteq V$ be a minimum $s,t$-cut in $G$ where $s,t \in T$ that has weight $\delta(S) \leq w$.
Then $G'$ preserves this cut $S$ (i.e., never contracts two nodes that are on different sides).
\end{claim}

The proof is similar to that of Claim~\ref{claim:monochromatic}, but the argument for lower bounding $\dem(L)$ is different (and does not rely on the friendliness of $S$ anymore).
There are two cases: either $x \in T$ or $x \notin T$.
If $x \in T$ we are in an easy case because $\dem(L)\geq \dem(x)=3 \phi^{-1}w$.
If $x \notin T$ then, since $S$ is a minimum $s,t$-cut for some pair of nodes where $x\neq s, x \neq t$ then $|E(\{x\},S)| \geq 0.5 \deg_G(x) \geq 0.4 \deg_G(x)$ or else moving $x$ to the other side gives a better $s,t$-cut.  From this, the proof can be carried on in the same way.
\end{proof}

\section{Conclusions}

In this paper we have put forward the notion of \emph{friendly cut sparsification} and proved that all friendly cuts with $\leq w$ edges in a simple graph can be preserved with a contracted graph on only $\tilde{O}(n \sqrt{w})$ edges.
This edge bound is tight up to log factors, and moreover, the sparsifier can be computed algorithmically in near-linear time.
Plugging this sparsification result instead of Nagamochi-Ibaraki's $O(nw)$-edge sparsifiers into the recent algorithms for \GHT (along with a new subroutine for single-source \emph{unfriendly} minimum cuts) leads to the new state-of-the-art bound of $\min\{ m+n^{1.75}, m \sqrt{n} \}\cdot n^{o(1)} $ for computing a \GHT of a simple graph.

Furthermore, we have shown that an $O(n \log n)$ edge bound is possible (for all $w$) if we only want to preserve friendly \emph{minimum} $s,t$-cuts in a simple graph, and that such a sparsifier can be computed in $(m+n^{1.5})^{1+o(1)}$ time \emph{if and only if} a \GHT can be computed in that time.
Since $\Omega(m+n^{1.5})$ is also a conditional lower bound for \GHT in simple graphs, unless the cubic barrier for multigraphs can be broken, this shows that better sparsification is the way towards tight bounds.
In a sense, the only remaining question for \GHT in simple graphs is that of computing better friendly sparsifiers.

It is plausible that the ``dynamic pivot'' derandomization technique of Abboud, Krauthgamer, and Trabelsi \cite{AKT21_focs} can also be combined with our sparsifiers, leading to a deterministic $(m+n^{1.75})^{1+o(1)}$ upper bound, assuming deterministic $m^{1+o(1)}$-time \MF.
And it is also plausible that the technique of Zhang \cite{Zhang21} for avoiding the $n^{o(1)}$ factors from using expander decomposition with node-demands could be incorporated to give a randomized $\tilde{O}(m+n^{1.75})$ algorithm, assuming $\tilde{O}(m)$-time \MF.

A side corollary of our results is a subcubic $n^{2.5+o(1)}$ algorihtm for single-source minimum-cuts in simple gaphs (and therefore also \GHT), that is considerably simpler than existing methods \cite{AKT21_stoc,AKT21_focs,LPS21_focs,Zhang21}:
execute the classical Gomory-Hu $n\cdot \MF(m)$ algorithm on the friendly $n$-cut sparsifier of Theorem~\ref{thm:friendly} with $m=\tilde{O}(n^{1.5})$ and use the algorithm of Theorem~\ref{thm:nonfriendly} for the unfriendly cuts.

{ \ifprocs
  \else
  \small
  \fi
\bibliographystyle{alphaurlinit}
\bibliography{robi}

\newcommand{\etalchar}[1]{$^{#1}$}
\begin{thebibliography}{vdBLL{\etalchar{+}}21}

\bibitem[AGU72]{AGU72}
A.~V. Aho, M.~R. Garey, and J.~D. Ullman.
\newblock The transitive reduction of a directed graph.
\newblock {\em SIAM Journal on Computing}, 1(2):131--137, 1972.
\newblock \href {http://dx.doi.org/10.1137/0201008}
  {\path{doi:10.1137/0201008}}.

\bibitem[AKT20a]{AKT20_b}
A.~Abboud, R.~Krauthgamer, and O.~Trabelsi.
\newblock Cut-equivalent trees are optimal for min-cut queries.
\newblock In {\em 61st {IEEE} Annual Symposium on Foundations of Computer
  Science}, FOCS'20, pages 105--118, 2020.
\newblock \href {http://dx.doi.org/10.1109/FOCS46700.2020.00019}
  {\path{doi:10.1109/FOCS46700.2020.00019}}.

\bibitem[AKT20b]{AKT20}
A.~Abboud, R.~Krauthgamer, and O.~Trabelsi.
\newblock New algorithms and lower bounds for all-pairs max-flow in undirected
  graphs.
\newblock In {\em Proceedings of the Thirty-First Annual ACM-SIAM Symposium on
  Discrete Algorithms}, SODA’20, page 48–61, 2020.
\newblock \href {http://dx.doi.org/10.1137/1.9781611975994.4}
  {\path{doi:10.1137/1.9781611975994.4}}.

\bibitem[AKT21a]{AKT21_focs}
A.~Abboud, R.~Krauthgamer, and O.~Trabelsi.
\newblock {APMF} $<$ {APSP}? {Gomory-Hu} tree for unweighted graphs in
  almost-quadratic time.
\newblock {\em Accepted to FOCS'21}, 2021.

\bibitem[AKT21b]{AKT21_stoc}
A.~Abboud, R.~Krauthgamer, and O.~Trabelsi.
\newblock Subcubic algorithms for {G}omory--{H}u tree in unweighted graphs.
\newblock In {\em Proceedings of the 53rd Annual ACM SIGACT Symposium on Theory
  of Computing}, pages 1725--1737, 2021.
\newblock \href {http://dx.doi.org/10.1145/3406325.3451073}
  {\path{doi:10.1145/3406325.3451073}}.

\bibitem[AW14]{AbboudW14}
A.~Abboud and V.~V. Williams.
\newblock Popular conjectures imply strong lower bounds for dynamic problems.
\newblock In {\em 55th {IEEE} Annual Symposium on Foundations of Computer
  Science}, FOCS'14, pages 434--443, 2014.
\newblock \href {http://dx.doi.org/10.1109/FOCS.2014.53}
  {\path{doi:10.1109/FOCS.2014.53}}.

\bibitem[BHKP07]{BHKP07}
A.~Bhalgat, R.~Hariharan, T.~Kavitha, and D.~Panigrahi.
\newblock An {$\tO(mn)$} {G}omory-{H}u tree construction algorithm for
  unweighted graphs.
\newblock In {\em 39th Annual ACM Symposium on Theory of Computing}, STOC'07,
  pages 605--614, 2007.
\newblock \href {http://dx.doi.org/10.1145/1250790.1250879}
  {\path{doi:10.1145/1250790.1250879}}.

\bibitem[BK15]{BeK15}
A.~A. Bencz{\'{u}}r and D.~R. Karger.
\newblock Randomized approximation schemes for cuts and flows in capacitated
  graphs.
\newblock {\em {SIAM} J. Comput.}, 44(2):290--319, 2015.
\newblock \href {http://dx.doi.org/10.1137/070705970}
  {\path{doi:10.1137/070705970}}.

\bibitem[BL16]{BanLinial16}
A.~Ban and N.~Linial.
\newblock Internal partitions of regular graphs.
\newblock {\em Journal of Graph Theory}, 83(1):5--18, 2016.

\bibitem[BPWZ14]{BPWZ14}
A.~Bj{\"{o}}rklund, R.~Pagh, V.~V. Williams, and U.~Zwick.
\newblock Listing triangles.
\newblock In {\em 41st International Colloquium on Automata, Languages, and
  Programming, {ICALP} 2014}, volume 8572 of {\em Lecture Notes in Computer
  Science}, pages 223--234. Springer, 2014.
\newblock \href {http://dx.doi.org/10.1007/978-3-662-43948-7\_19}
  {\path{doi:10.1007/978-3-662-43948-7\_19}}.

\bibitem[BTV06]{BTV06}
C.~Bazgan, Z.~Tuza, and D.~Vanderpooten.
\newblock The satisfactory partition problem.
\newblock {\em Discrete Applied Mathematics}, 154(8):1236--1245, 2006.

\bibitem[CGL{\etalchar{+}}20]{CGLNPS20}
J.~Chuzhoy, Y.~Gao, J.~Li, D.~Nanongkai, R.~Peng, and T.~Saranurak.
\newblock A deterministic algorithm for balanced cut with applications to
  dynamic connectivity, flows, and beyond.
\newblock In {\em IEEE 61st Annual Symposium on Foundations of Computer Science
  (FOCS)}, pages 1158--1167, 2020.
\newblock \href {http://dx.doi.org/10.1109/FOCS46700.2020.00111}
  {\path{doi:10.1109/FOCS46700.2020.00111}}.

\bibitem[CQ21]{CQ21}
C.~Chekuri and K.~Quanrud.
\newblock Isolating cuts,(bi-)submodularity, and faster algorithms for
  connectivity.
\newblock In {\em 48th International Colloquium on Automata, Languages, and
  Programming (ICALP 2021)}. Schloss Dagstuhl-Leibniz-Zentrum f{\"u}r
  Informatik, 2021.

\bibitem[DG19]{DudekG19}
B.~Dudek and P.~Gawrychowski.
\newblock Computing quartet distance is equivalent to counting 4-cycles.
\newblock In {\em Proceedings of the 51st Annual {ACM} {SIGACT} Symposium on
  Theory of Computing}, STOC'19, pages 733--743. {ACM}, 2019.
\newblock \href {http://dx.doi.org/10.1145/3313276.3316390}
  {\path{doi:10.1145/3313276.3316390}}.

\bibitem[DKW15]{DKW15}
M.~Dinitz, R.~Krauthgamer, and T.~Wagner.
\newblock Towards resistance sparsifiers.
\newblock In {\em Approximation, Randomization, and Combinatorial Optimization.
  Algorithms and Techniques (APPROX/RANDOM 2015)}, volume~40 of {\em Leibniz
  International Proceedings in Informatics (LIPIcs)}, pages 738--755. Schloss
  Dagstuhl--Leibniz-Zentrum fuer Informatik, 2015.
\newblock \href {http://dx.doi.org/10.4230/LIPIcs.APPROX-RANDOM.2015.738}
  {\path{doi:10.4230/LIPIcs.APPROX-RANDOM.2015.738}}.

\bibitem[FKN{\etalchar{+}}21]{Ferber21}
A.~Ferber, M.~Kwan, B.~Narayanan, A.~Sah, and M.~Sawhney.
\newblock Friendly bisections of random graphs.
\newblock {\em arXiv preprint arXiv:2105.13337}, 2021.

\bibitem[Gab95]{Gabow95}
H.~N. Gabow.
\newblock A matroid approach to finding edge connectivity and packing
  arborescences.
\newblock {\em Journal of Computer and System Sciences}, 50(2):259--273, 1995.

\bibitem[GH61]{GH61}
R.~E. Gomory and T.~C. Hu.
\newblock Multi-terminal network flows.
\newblock {\em Journal of the Society for Industrial and Applied Mathematics},
  9:551--570, 1961.
\newblock Available from: \url{http://www.jstor.org/stable/2098881}.

\bibitem[GK00]{GerberKobler00}
M.~U. Gerber and D.~Kobler.
\newblock Algorithmic approach to the satisfactory graph partitioning problem.
\newblock {\em European Journal of Operational Research}, 125(2):283--291,
  2000.

\bibitem[GMT20]{GMT20}
A.~Gaikwad, S.~Maity, and S.~K. Tripathi.
\newblock The satisfactory partition problem.
\newblock {\em arXiv preprint arXiv:2007.14339}, 2020.

\bibitem[Gup01]{Gupta01}
A.~Gupta.
\newblock Steiner points in tree metrics don't (really) help.
\newblock In {\em Proceedings of the 12th Annual Symposium on Discrete
  Algorithms}, pages 220--227, 2001.
\newblock Available from:
  \url{http://dl.acm.org/citation.cfm?id=365411.365448}.

\bibitem[HKNR98]{HKNR98}
T.~Hagerup, J.~Katajainen, N.~Nishimura, and P.~Ragde.
\newblock Characterizing multiterminal flow networks and computing flows in
  networks of small treewidth.
\newblock {\em J. Comput. Syst. Sci.}, 57:366--375, 1998.
\newblock \href {http://dx.doi.org/10.1006/jcss.1998.1592}
  {\path{doi:10.1006/jcss.1998.1592}}.

\bibitem[JRR95]{junger1995traveling}
M.~J{\"u}nger, G.~Reinelt, and G.~Rinaldi.
\newblock The traveling salesman problem.
\newblock {\em Handbooks in Operations Research and Management Science},
  7:225--330, 1995.

\bibitem[Kar99]{Karger99}
D.~R. Karger.
\newblock Random sampling in cut, flow, and network design problems.
\newblock {\em Mathematics of Operations Research}, 24(2):383--413, 1999.
\newblock \href {http://dx.doi.org/10.1287/moor.24.2.383}
  {\path{doi:10.1287/moor.24.2.383}}.

\bibitem[KKR12]{KKR12}
K.-i. Kawarabayashi, Y.~Kobayashi, and B.~Reed.
\newblock The disjoint paths problem in quadratic time.
\newblock {\em Journal of Combinatorial Theory, Series B}, 102(2):424--435,
  2012.

\bibitem[KNZ14]{KNZ14}
R.~Krauthgamer, H.~Nguyen, and T.~Zondiner.
\newblock Preserving terminal distances using minors.
\newblock {\em SIAM Journal on Discrete Mathematics}, 28(1):127--141, 2014.
\newblock \href {http://dx.doi.org/10.1137/120888843}
  {\path{doi:10.1137/120888843}}.

\bibitem[KPP16]{KPP16}
T.~Kopelowitz, S.~Pettie, and E.~Porat.
\newblock Higher lower bounds from the {3SUM} conjecture.
\newblock In {\em Proceedings of the twenty-seventh annual ACM-SIAM symposium
  on Discrete algorithms}, pages 1272--1287. SIAM, 2016.

\bibitem[KT19]{KThorup19}
K.~Kawarabayashi and M.~Thorup.
\newblock Deterministic edge connectivity in near-linear time.
\newblock {\em J. {ACM}}, 66(1):4:1--4:50, 2019.
\newblock \href {http://dx.doi.org/10.1145/3274663}
  {\path{doi:10.1145/3274663}}.

\bibitem[LNP{\etalchar{+}}21]{LNPSS21}
J.~Li, D.~Nanongkai, D.~Panigrahi, T.~Saranurak, and S.~Yingchareonthawornchai.
\newblock Vertex connectivity in poly-logarithmic max-flows.
\newblock In {\em Proceedings of the 53rd Annual ACM SIGACT Symposium on Theory
  of Computing}, pages 317--329, 2021.

\bibitem[LP20]{LP20}
J.~Li and D.~Panigrahi.
\newblock Deterministic min-cut in poly-logarithmic max-flows.
\newblock In {\em 61st {IEEE} Annual Symposium on Foundations of Computer
  Science}, FOCS'20, pages 85--92, 2020.
\newblock \href {http://dx.doi.org/10.1109/FOCS46700.2020.00017}
  {\path{doi:10.1109/FOCS46700.2020.00017}}.

\bibitem[LP21]{LP21}
J.~Li and D.~Panigrahi.
\newblock Approximate {G}omory-{H}u tree is faster than $n-1$ max-flows.
\newblock In {\em 53rd Annual {ACM} {SIGACT} Symposium on Theory of Computing},
  STOC'21, pages 1738--1748. {ACM}, 2021.
\newblock \href {http://dx.doi.org/10.1145/3406325.3451112}
  {\path{doi:10.1145/3406325.3451112}}.

\bibitem[LPS21]{LPS21_focs}
J.~Li, D.~Panigrahi, and T.~Saranurak.
\newblock A nearly optimal all-pairs min-cuts algorithm in simple graphs.
\newblock {\em Accepted to FOCS'21}, 2021.

\bibitem[LS21]{LS21}
J.~Li and T.~Saranurak.
\newblock Deterministic weighted expander decomposition in almost-linear time.
\newblock {\em CoRR}, abs/2106.01567, 2021.
\newblock \href {http://arxiv.org/abs/2106.01567} {\path{arXiv:2106.01567}}.

\bibitem[LSDK18]{LSDK18}
Y.~Liu, T.~Safavi, A.~Dighe, and D.~Koutra.
\newblock Graph summarization methods and applications: A survey.
\newblock {\em ACM Comput. Surv.}, 51(3), 2018.
\newblock \href {http://dx.doi.org/10.1145/3186727}
  {\path{doi:10.1145/3186727}}.

\bibitem[MN21]{MN21}
S.~Mukhopadhyay and D.~Nanongkai.
\newblock A note on isolating cut lemma for submodular function minimization.
\newblock {\em arXiv preprint arXiv:2103.15724}, 2021.

\bibitem[Moi09]{Moitra09}
A.~Moitra.
\newblock Approximation algorithms for multicommodity-type problems with
  guarantees independent of the graph size.
\newblock In {\em Proceedings of the 50th Annual IEEE Symposium on Foundations
  of Computer Science}, FOCS'09, page 3–12. IEEE Computer Society, 2009.
\newblock \href {http://dx.doi.org/10.1109/FOCS.2009.28}
  {\path{doi:10.1109/FOCS.2009.28}}.

\bibitem[Mor00]{Morris00}
S.~Morris.
\newblock Contagion.
\newblock {\em The Review of Economic Studies}, 67(1):57--78, 2000.

\bibitem[NI92]{NI92}
H.~Nagamochi and T.~Ibaraki.
\newblock Linear time algorithms for finding $k$-edge connected and $k$-node
  connected spanning subgraphs.
\newblock {\em Algorithmica}, 7:583--596, 1992.
\newblock \href {http://dx.doi.org/10.1007/BF01758778}
  {\path{doi:10.1007/BF01758778}}.

\bibitem[Pat10]{Pat10}
M.~Patrascu.
\newblock Towards polynomial lower bounds for dynamic problems.
\newblock In {\em Proceedings of the forty-second ACM symposium on Theory of
  computing}, pages 603--610, 2010.

\bibitem[PS89]{PS89}
D.~Peleg and A.~A. Sch{\"a}ffer.
\newblock Graph spanners.
\newblock {\em J. Graph Theory}, 13(1):99--116, 1989.
\newblock \href {http://dx.doi.org/10.1002/jgt.3190130114}
  {\path{doi:10.1002/jgt.3190130114}}.

\bibitem[R\"02]{Raecke02}
H.~R\"{a}cke.
\newblock Minimizing congestion in general networks.
\newblock In {\em 43rd Symposium on Foundations of Computer Science}, pages
  43--52. IEEE Computer Society, 2002.
\newblock \href {http://dx.doi.org/10.1109/SFCS.2002.1181881}
  {\path{doi:10.1109/SFCS.2002.1181881}}.

\bibitem[ST83]{SleatorT83}
D.~D. Sleator and R.~E. Tarjan.
\newblock A data structure for dynamic trees.
\newblock {\em J. Comput. Syst. Sci.}, 26(3):362--391, 1983.
\newblock \href {http://dx.doi.org/10.1016/0022-0000(83)90006-5}
  {\path{doi:10.1016/0022-0000(83)90006-5}}.

\bibitem[ST11]{ST11}
D.~A. Spielman and S.-H. Teng.
\newblock Spectral sparsification of graphs.
\newblock {\em SIAM J. Comput.}, 40(4):981--1025, 2011.
\newblock \href {http://dx.doi.org/10.1137/08074489X}
  {\path{doi:10.1137/08074489X}}.

\bibitem[SW97]{SW97}
M.~Stoer and F.~Wagner.
\newblock A simple min-cut algorithm.
\newblock {\em Journal of the ACM (JACM)}, 44(4):585--591, 1997.

\bibitem[SW19]{SW19}
T.~Saranurak and D.~Wang.
\newblock Expander decomposition and pruning: Faster, stronger, and simpler.
\newblock In {\em Proceedings of the 30th Annual {ACM-SIAM} Symposium on
  Discrete Algorithms}, SODA'19, pages 2616--2635, 2019.
\newblock \href {http://dx.doi.org/10.1137/1.9781611975482.162}
  {\path{doi:10.1137/1.9781611975482.162}}.

\bibitem[vdBLL{\etalchar{+}}21]{linearflow21}
J.~van~den Brand, Y.~T. Lee, Y.~P. Liu, T.~Saranurak, A.~Sidford, Z.~Song, and
  D.~Wang.
\newblock Minimum cost flows, {MDPs}, and $\ell_1$-regression in nearly linear
  time for dense instances.
\newblock In {\em 53rd Annual ACM SIGACT Symposium on Theory of Computing},
  STOC'21, page 859–869. ACM, 2021.
\newblock \href {http://dx.doi.org/10.1145/3406325.3451108}
  {\path{doi:10.1145/3406325.3451108}}.

\bibitem[YZ97]{YZ97}
R.~Yuster and U.~Zwick.
\newblock Finding even cycles even faster.
\newblock {\em SIAM Journal on Discrete Mathematics}, 10(2):209--222, 1997.

\bibitem[Zha21]{Zhang21}
T.~Zhang.
\newblock Faster cut-equivalent trees in simple graphs.
\newblock {\em arXiv preprint arXiv:2106.03305}, 2021.

\end{thebibliography}
}

\end{document}